\documentclass[10pt]{article}
\pdfoutput=1

\usepackage{jheppub-mod}
\usepackage{amssymb}
\usepackage{amsmath}
\usepackage{amsthm}
\usepackage{mathtools}
\usepackage[usenames,dvipsnames]{xcolor}
\usepackage{epsfig}
\usepackage{dcolumn}
\usepackage{tikz}
\usepackage{tikz-cd}
\usetikzlibrary{shapes.geometric, arrows}
\usepackage{upgreek}
\usepackage{setspace}
\usepackage{enumitem}
\usepackage{array,multirow,bigdelim,arydshln}
\usepackage{appendix}
\usepackage[export]{adjustbox}
\usepackage{xparse}
\usepackage[utf8]{inputenc}
\usepackage{microtype}

\def \be {\begin{equation}}
\def \ee {\end{equation}}
\def \tr {\\}

\def \nn {\nonumber}
\def \la {\langle}
\def \ra {\rangle}
\def \L {\mathcal{L}_{\omega}}

\def \con {\nabla_\omega}

\def \I {\mathbb{I}}
\def \w {\,\wedge\,}
\def \C {\mathsf{C}}
\def \tC {\widetilde{\mathsf{C}}}
\def \reg {\text{reg}}
\def \M {\mathcal{M}}
\def \PT {\mathsf{PT}}
\DeclareMathOperator{\sgn}{sgn}
\newcommand{\overbar}[1]{\mkern 1.5mu\overline{\mkern-1.5mu#1\mkern-1.5mu}\mkern 1.5mu}

\mathtoolsset{showonlyrefs}
\newtheorem{theorem}{Theorem}[section]
\newtheorem{claim}{Claim}[section]
\newtheorem{lemma}{Lemma}[section]
\newtheorem*{remark}{Remark}
\newtheorem{definition}{Definition}[section]
\newtheorem{example}{Example}[section]
\newtheorem*{acknowledgements}{Acknowledgements}
\newtheorem*{outline}{Outline}

\allowdisplaybreaks
\graphicspath{{figures/}}
\renewcommand{\thefigure}{\thesection.\arabic{figure}}

\title{Combinatorics {\itshape\LARGE and} Topology\\ {\itshape\LARGE of} Kawai--Lewellen--Tye Relations}
\author{Sebastian Mizera}
\affiliation{Perimeter Institute for Theoretical Physics, Waterloo, ON N2L 2Y5, Canada}
\affiliation{Department of Physics \& Astronomy, University of Waterloo, Waterloo, ON N2L 3G1, Canada}
\emailAdd{smizera@pitp.ca}
\abstract{We revisit the relations between open and closed string scattering amplitudes discovered by Kawai, Lewellen, and Tye (KLT). We show that they emerge from the underlying algebro-topological identities known as the \emph{twisted period relations}. In order to do so, we formulate tree-level string theory amplitudes in the language of \emph{twisted de Rham theory}. There, open string amplitudes are understood as pairings between \emph{twisted cycles} and \emph{cocycles}. Similarly, closed string amplitudes are given as a pairing between two twisted cocycles. Finally, objects relating the two types of string amplitudes are the $\alpha'$-corrected bi-adjoint scalar amplitudes recently defined by the author \cite{Mizera:2016jhj}. We show that they naturally arise as \emph{intersection numbers} of twisted cycles. In this work we focus on the combinatorial and topological description of twisted cycles relevant for string theory amplitudes. In this setting, each twisted cycle is a polytope, known in combinatorics as the \emph{associahedron}, together with an additional structure encoding monodromy properties of string integrals. In fact, this additional structure is given by higher-dimensional generalizations of the Pochhammer contour. An open string amplitude is then computed as an integral of a logarithmic form over an associahedron. We show that the inverse of the KLT kernel can be calculated from the knowledge of how pairs of associahedra intersect one another in the moduli space. In the field theory limit, contributions from these intersections localize to vertices of the associahedra, giving rise to the bi-adjoint scalar partial amplitudes.}

\begin{document}

\maketitle
\setcounter{page}{2}

\vfill
\begin{acknowledgements}
We thank Freddy Cachazo for insightful comments on this work. We are grateful to Oliver Schlotterer for bringing the paper \cite{Mimachi2004} to our attention, which has inspired this line of research. We also thank Paolo Benincasa, Lauren Williams, and Karen Yeats for useful discussions. This research was supported in part by Perimeter Institute for Theoretical Physics. Research at Perimeter Institute is supported by the Government of Canada through the Department of Innovation, Science and Economic Development Canada and by the Province of Ontario through the Ministry of Research, Innovation and Science.
\end{acknowledgements}

\pagebreak
\section{Introduction}

\noindent\textsc{Recent years have seen} a vast improvement in our understanding of quantum field theories through the study of scattering amplitudes \cite{Elvang:2013cua}. Such advancements were often made possible by considering a generalization of ordinary field theories into string theories. The main advantage of this approach is that strings---as extended objects---provide a way of smoothing out interactions between the scattering states. More precisely, the moduli space of a string worldsheet continuously connects its different factorization channels. As a result, a sum over discrete objects---such as Feynman \cite{PhysRev.76.769} or on-shell \cite{ArkaniHamed:2012nw} diagrams---in field theory is replaced by an integral over a continuous worldsheet in string theory. In the infinite tension limit, where strings become point-like, this integral localizes to disconnected corners of the moduli space, which give rise to the field theory amplitudes. In this way, thinking of field theory amplitudes as a limit of the string theory ones provides a way of unifying all factorization channels under a single object.

The prime example of usefulness of string theory in the study of field theory amplitudes are the Kawai--Lewellen--Tye (KLT) relations discovered in 1985 \cite{Kawai:1985xq}. They give a way of writing the amplitudes for scattering of closed strings entirely in terms of a quadratic combination of open string amplitudes. In the field theory limit, where closed strings reduce to gravitons---particle excitations of General Relativity---and open strings reduce to gluons---excitations of the Yang--Mills theory---KLT relations give a connection between graviton and gluon scattering amplitudes. Such a relationship not only hints at a fundamental interplay between the two types of theories, but also provides enormous simplifications for practical calculations, both in string and field theory.

KLT relations have been most thoroughly studied in the field theory limit. In its modern form found by Cachazo, He, and Yuan (CHY) they read \cite{Cachazo:2013iea}:
\be\label{intro-KLT-FT}
\mathcal{A}^{\text{GR}} \;=\; \sum_{\beta, \gamma}\; \mathcal{A}^{\text{YM}}(\beta)\; m^{-1}(\beta | \gamma)\; \mathcal{A}^{\text{YM}}(\gamma).
\ee
Here, $\mathcal{A}^{\text{GR}}$ is an $n$-point graviton amplitude, while $\mathcal{A}^{\text{YM}}(\beta)$ is an $n$-point gluon partial amplitude with ordering $\beta$. The sum proceeds over two sets of $(n-3)!$ permutations $\beta$ and $\gamma$ forming a basis for the Yang--Mills amplitudes. The object $m(\beta | \gamma)$ is a double-partial amplitude of a bi-adjoint scalar theory \cite{BjerrumBohr:2012mg,Cachazo:2013iea}. It is convenient to think of the relation \eqref{intro-KLT-FT} as a matrix product of a transposed vector, inverse of a matrix, and another vector, where rows and columns are labelled by permutations.

It was not always clear that coefficients of the KLT expansion can be written in the form \eqref{intro-KLT-FT} as the inverse of a matrix. In their original work, Kawai, Lewellen, and Tye used contour deformation arguments to arrive at these coefficients as coming from monodromy factors around vertex operators on the boundary of a worldsheet \cite{Kawai:1985xq}. They evaluated explicit form of the quadratic relations for low-point examples. A closed-form expression for the KLT relations to arbitrary number of particles in field theory was later given in Appendix~A of \cite{Bern:1998sv} by Bern, Dixon, Perelstein, and Rozowsky. Properties of this expansion were systematically studied and proven in a series of papers \cite{BjerrumBohr:2010ta,BjerrumBohr:2010zb,BjerrumBohr:2010yc,BjerrumBohr:2010hn} by Bjerrum-Bohr, Damgaard, Feng, S{\o}ndergaard, and Vanhove, who also generalized the allowed bases of permutations to a larger set. They introduced the matrix $S[\beta|\gamma]$ called a \emph{KLT kernel}, which allows to write the KLT relations as a matrix product. Finally, Cachazo, He, and Yuan recognized \cite{Cachazo:2013iea} that the KLT kernel can be understood as the inverse matrix of bi-adjoint scalar amplitudes, i.e., $S[\beta | \gamma] = m^{-1}(\beta | \gamma)$, ultimately leading to the form given in \eqref{intro-KLT-FT}. This also allowed to construct the kernel from the most general sets of permutations labelling the columns and rows of $m(\beta | \gamma)$, so that coefficients of the KLT expansion are not necessarily polynomials in the kinematic invariants.

At this point one could ask: \emph{Where do KLT relations come from?} It turns out that a fruitful path to consider is to go back to the string theory case, where these relations were first conceived. It was proposed by the author \cite{Mizera:2016jhj} that the full string theory KLT relations can be rewritten in a form analogous to \eqref{intro-KLT-FT} as follows:
\be\label{intro-KLT-ST}
\mathcal{A}^{\text{closed}} \;=\; \sum_{\beta, \gamma}\; \mathcal{A}^{\text{open}}(\beta)\; m_{\alpha'}^{-1}(\beta | \gamma)\; \mathcal{A}^{\text{open}}(\gamma).
\ee
Here, $\mathcal{A}^{\text{closed}}$ and $\mathcal{A}^{\text{open}}(\beta)$ are the $n$-point closed and open string amplitudes respectively. The role of the string theory KLT kernel is played by the inverse of a matrix $m_{\alpha'}(\beta | \gamma)$, which is constructed out of the bi-adjoint scalar amplitudes with $\alpha'$ corrections. Recall that $\alpha'$ is a parameter inversely proportional to the string tension, such that $\alpha' \to 0$ corresponds to the field theory limit. In this way, \eqref{intro-KLT-ST} is a direct analogue of \eqref{intro-KLT-FT}, where every piece of the puzzle receives string corrections. By evaluating explicit examples of $m_{\alpha'}(\beta | \gamma)$, which from now on we will refer to as the \emph{inverse KLT kernel}, we found that they have a surprisingly simple structure, giving rise to compact expressions in terms of trigonometric functions. Moreover, they can be calculated using Feynman-like diagrammatic rules \cite{Mizera:2016jhj}, hinting at an underlying combinatorial underpinnings. In this work we show that string theory KLT relations in the form \eqref{intro-KLT-ST} are in fact a result of a deep connection between string theory amplitudes, algebraic topology, and combinatorics.

Practically at the same time as the initial work on the KLT relations, on the other side of the globe, mathematicians Aomoto, Cho, Kita, Matsumoto, Mimachi, Yoshida, and collaborators were developing a seemingly unrelated theory of hypergeometric functions \cite{aomoto2011theory,yoshida2013hypergeometric}. It eventually led to the formulation of \emph{twisted de Rham theory}, which is a generalization of the conventional de Rham theory to integrals of multi-valued functions \cite{aomoto2011theory}. Let us first intuitively explain its key ingredients, leaving precise definitions for later sections. A twisted homology group $H_m(X,\mathcal{L}_\omega)$ on some manifold $X$ is a space of \emph{twisted cycles}, which are regions of $X$ together with an additional information about branches of a multi-valued function. Similarly, a twisted cohomology group $H^m(X,\nabla_{\omega})$ is a space of \emph{twisted cocycles}, which are differential forms on $X$ satisfying certain conditions. A pairing between a twisted cycle and a cocycle is then simply an integral of a differential form over a given region of $X$ which is sensitive to the branch structure of the integrand. \emph{Twist} measures multi-valuedness of the integrand.

One can also define a natural set of a \emph{dual} twisted homology $H_m(X,\mathcal{L}^\vee_\omega)$ and a \emph{dual} twisted cohomology $H^m(X,\nabla_{\omega}^\vee)$. For the purpose of this work, the duality is roughly speaking given by complex conjugation. One can define a pairing between these two dual spaces too, giving rise to another integral of a multi-valued function. Having defined two different pairs of twisted homologies and cohomologies, we would like to calculate invariants between them as well. As it turns out, it is possible to pair two twisted cycles belong to a twisted homology and its dual. The resulting object is called an \emph{intersection number of twisted cycles} \cite{MANA:MANA19941660122,MANA:MANA19941680111,MANA:MANA173,MANA:MANA200310105,Matsumoto-Yoshida-progress}. It is computed from the information of how these cycles intersect one another in $X$, as well as their associated branch structure. Similarly, one can also define an \emph{intersection number of twisted cocycles} \cite{cho1995}. What is more, in 1994 Cho and Matsumoto found identities---known as the \emph{twisted period relations}---between pairings computed from different twisted homologies and cohomologies described above \cite{cho1995}.

In this work we show that Kawai--Lewellen--Tye relations are a consequence of twisted period relations. In order to do so, we first formulate string theory tree-level amplitudes in the language of twisted de Rham theory. Open string partial amplitudes $\mathcal{A}^{\text{open}}(\beta)$ are given as pairings between twisted cycles and twisted cocycles, while closed string amplitudes $\mathcal{A}^{\text{closed}}$ come from intersection numbers of twisted cocycles. Finally, inverse of the KLT kernel $m_{\alpha'}(\beta | \gamma)$ is calculated from intersection numbers of twisted cycles. We can schematically summarize these pairings in the following diagram:
\be\label{intro-diagram}
\begin{tikzcd}[row sep = 6em, column sep = 10em]
	H^m(X, \nabla_\omega) \arrow[leftrightarrow]{r}{\displaystyle \mathcal{A}^{\text{closed}}} \arrow[leftrightarrow]{d}[swap]{\displaystyle \mathcal{A}^{\text{open}}(\beta)}
	& H^m(X,\nabla^\vee_\omega) \arrow[leftrightarrow]{d}{\;\displaystyle \mathcal{A}^{\text{open}}(\gamma)} \\
	H_m(X,{\cal L}_\omega) \arrow[leftrightarrow]{r}[swap]{\displaystyle m_{\alpha'}(\beta | \gamma)}
	& H_m(X,{\cal L}^\vee_\omega)
\end{tikzcd}
\ee
Twisted period relations for the above pairings become KLT relations in exactly the same form as \eqref{intro-KLT-ST}. We give a proof of this statement in Section~\ref{sec-klt-as-twisted-period-relations}, where we also define bases of twisted cycles and cocycles relevant for string amplitudes.

These twisted cycles and cocycles turn out to have interesting combinatorial properties. It is known that an $n$-point tree-level open string partial amplitude is given by an integral of a differential form over a simplex $\Delta_{n-3}$, belonging to the moduli space of genus-zero Riemann surfaces with $n$ punctures \cite{green1988superstring}. However, in order to resolve degenerate points close to the vertices of the simplex, one considers a \emph{blowup} of the moduli space, $\pi^{-1}(\mathcal{M}_{0,n}) = \widetilde{\mathcal{M}}_{0,n}$ \cite{Deligne1969,MathScand11642,MathScand12001,MathScand12002,DeConcini1995}. On this space, the simplex becomes a different polytope known as the \emph{associahedron}, $K_{n-1}$ \cite{Devadoss98tessellationsof}. An example of this procedure is given below: 
\be
\includegraphics{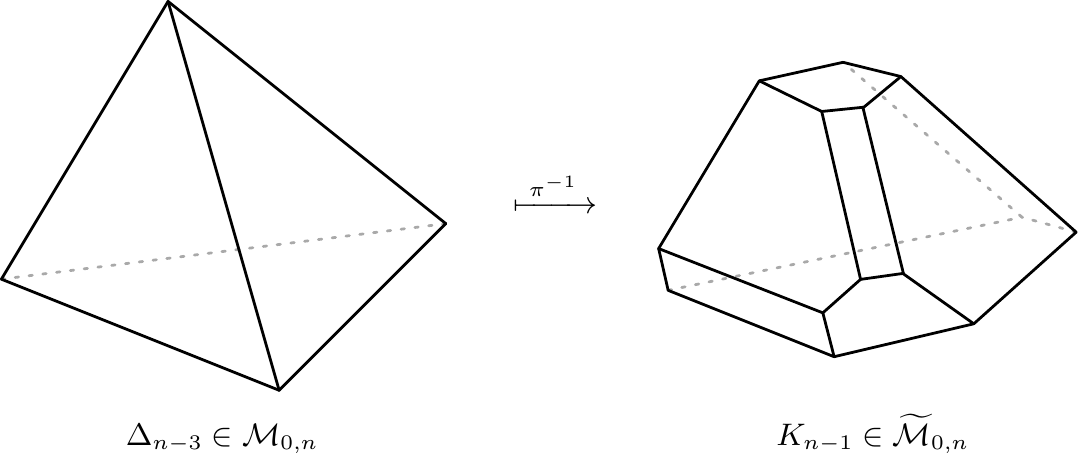}
\ee
Twisted cycles are then given by $(n-3)$-dimensional associahedra with an additional structure keeping track of the branches of the integrand. This structure is most conveniently summarized by introducing an additional regularization of twisted cycles based on the Pochhammer contour \cite{Witten:2013pra} and its higher-dimensional generalizations. We give details of this construction in Section~\ref{subsec-regularization}. In Section~\ref{subsec-twisted-cocycles} we also find a basis of twisted cocycles for string amplitudes and show they are given by logarithmic $(n-3)$-forms. With these constructions, an open string partial amplitude becomes an integration of a logarithmic form over an associahedron. It is interesting to see how physical properties arise in this formulation. \emph{Unitarity} is made manifest from the fact that facets of the associahedra are given by products of two lower-dimensional associahedra. \emph{Locality} is manifest from the fact that a higher-dimensional Pochhammer contour yields only simple poles in all factorization channels. Similarly, each lower-dimensional face of the associahedron has an associated factorization diagram. For instance, contact terms come from the bulk of the polytope, while trivalent diagrams come from its vertices. Since each propagator comes with a power of $\alpha'$, it means that in the field theory limit only the regions of the moduli space around the vertices of the associahedra contribute. We show how to construct this limit explicitly in Appendix~\ref{app-field-theory-limit}.

The most novel concept studied in this work, however, is the evaluation of the intersection numbers of twisted cycles. We show how to calculate them on explicit examples and in general in Section~\ref{sec-inverse-klt}. There, we also prove that combinatorial rules for finding intersection numbers are equivalent to the diagrammatic expansion found empirically in \cite{Mizera:2016jhj}, establishing that the $\alpha'$-corrected bi-adjoint scalar amplitudes $m_{\alpha'}(\beta | \gamma)$ are given by intersection numbers of twisted cycles. Geometric and topological meaning of these objects can be easily pictured. The real section of the moduli space $\widetilde{\mathcal{M}}_{0,n}$ is tiled by $(n-1)!/2$ associahedra $K_{n-1}(\beta)$ \cite{Devadoss98tessellationsof,Devadoss_combinatorialequivalence}, each labelled with some permutation $\beta$. The problem of calculating $m_{\alpha'}(\beta | \gamma)$ reduces to finding the intersection of two associahedra $K_{n-1}(\beta)$ and $K_{n-1}(\gamma)$ in the moduli space:
\be
\includegraphics{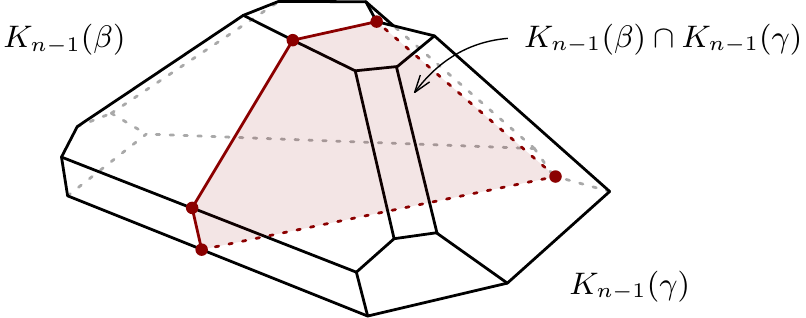}
\ee
The intersection number then receives contributions from all the $(0,1,2,\ldots)$-dimensional faces belonging to the intersection $K_{n-1}(\beta) \cap K_{n-1}(\gamma)$. In the above example, these are five vertices, five edges, and one polygon. Once again, in the field theory limit these contributions localize to vertices only, and hence can be written as a sum over trivalent diagrams. Since the intersection region belongs to both associahedra $K_{n-1}(\beta)$ and $K_{n-1}(\gamma)$ at the same time, the trivalent diagrams have to be compatible with both planar orderings $\beta$ and $\gamma$. This is indeed the standard definition of the field theory bi-adjoint scalar double-partial amplitude $m(\beta | \gamma)$. It is quite surprising that a scattering amplitude in a quantum field theory can be understood as arising from such an abstract mathematical object as an intersection number of twisted cycles.

\begin{outline}
This paper is structured as follows. In Section~\ref{sec-preliminaries} we give an introduction to the topics of twisted de Rham theory, as well as string theory amplitudes. In Section~\ref{sec-klt-as-twisted-period-relations} we define the twisted cycles and cocycles that are relevant for string theory amplitude computations. There, we also establish the equivalence between Kawai--Lewellen--Tye relations and twisted period relations. In Section~\ref{sec-inverse-klt} we discuss the interpretation of the inverse KLT kernel as intersection numbers of twisted cycles. We give a combinatorial description of the blowup procedure leading to the associahedron, and present the regularization of twisted cycles using a generalized Pochhammer contour. After giving explicit examples of the evaluation of intersection numbers for lower-point cases, we prove that they are equivalent to the diagrammatic rules for the computation of $m_{\alpha'}(\beta | \gamma)$ in general. We conclude with the summary of the results and a discussion of open questions in Section~\ref{sec-conclusion}. In Appendix~\ref{app-field-theory-limit} we discuss how to obtain the field theory limit of open string amplitudes from contributions localized around the vertices of the associahedra.
\end{outline}

\pagebreak
\section{\label{sec-preliminaries}Mathematical \emph{\&} Physical Preliminaries}

\textsc{This section is meant} to give an informal introduction to both mathematics of twisted de Rham theory and physics of string theory amplitudes for the readers not familiar with these topics.

\subsection{Twisted de Rham Theory}

\textsc{In the study of hypergeometric functions} one encounters integrals of multi-valued functions. In order to analyze properties of such objects, it is useful to reformulate the problem in the language of algebraic topology, where integrals are understood as pairings between integration \emph{cycles} and corresponding \emph{cocycles} as the integrands. In the case when the integrand is a single-valued object, the problem is governed by \emph{de Rham theory} and its homology and cohomology groups \cite{hatcher2002algebraic}. In the case of multi-valued integrands, one needs to keep track of additional information about the branch structure along the integration region. Study of such objects leads to a generalization of de Rham theory into its \emph{twisted} version.

Twisted de Rham theory dates back to the work of Aomoto \cite{AOMK:1972,aomoto1974,aomoto1974b,aomoto1975,aomoto1982,aomoto1983}, Deligne \cite{deligne1970equations}, Kita \cite{kita1982,kita1992,kita1993,kita_1994}, and Gelfand \cite{Gelfand1986,GelfandGelfand1986} who laid out foundations for this theory, which later grew into a field of research developed by various authors, see, e.g., \cite{cho1995,MANA:MANA19941660122,MANA:MANA19941680111,MANA:MANA173,MANA:MANA200310105,Matsumoto-Yoshida-progress,KITA1997,aomoto1997twisted}. Overview of these results is presented in textbooks by Aomoto and Kita \cite{aomoto2011theory}, as well as Yoshida \cite{yoshida2013hypergeometric}.\footnote{See also textbooks by Haraoka \cite{Haraoka} and Kimura \cite{Kimura} in Japanese, as well as one by Orlik and Terao \cite{orlik1992arrangements}, who discuss hypergeometric functions from the viewpoint of arrangements of hyperplanes.} In this section we outline the basics of twisted de Rham theory that should serve as intuition for the remainder of the paper. We follow the discussion in \cite{aomoto2011theory}.

Despite initial motivation coming from hypergeometric functions, twisted de Rham theory extends to more general objects. We will consider integrals of the form:
\be\label{int-multi-valued}
\int_{\gamma} u(z)\, \varphi(z),
\ee
where $u(z)$ and $\varphi(z)$ are a multi-valued function and a single-valued differential form respectively. Let us define the function $u(z)$ as
\be
u(z) := \prod_{i=1}^k f_i(z)^{\alpha_i} \qquad\text{with}\qquad \alpha_i \in \mathbb{C} \setminus \mathbb{Z},
\ee
where $f_i(z) = f_i(z_1,z_2,\ldots,z_m)$ are linear polynomials defined on an $m$-dimensional complex space minus the singular locus of $u(z)$, called a divisor, $D$:
\be
X := \mathbb{C}^m \!\setminus\! D \qquad\text{with}\qquad D := \bigcup_{i=1}^{k} \,\{ f_i(z) = 0 \}.
\ee
The function $u(z)$ and the $m$-form $\varphi(z)$, together with the $m$-dimensional region $\gamma$ are defined on the same manifold $X$. We demand that $\gamma$ has endpoints only on the divisor $D$, which implies that it does not have any boundaries on $X$. Hence, $\gamma$ can be called a \emph{topological cycle}.

In order to give a more precise definition of \eqref{int-multi-valued} let us introduce a smooth triangulation of $X$ that will serve as an intuitive example. We take the cycle $\gamma$ to be an $m$-simplex $\Delta$. Since $u(z)$ is multi-valued, we need to specify its branch on $\Delta$. We use the notation $\Delta \otimes u_\Delta(z)$ to signify the choice of a branch $u_{\Delta}(z)$ of $u(z)$ on $\Delta$. With this definition \eqref{int-multi-valued} becomes:
\be
\int_{\Delta \otimes u_{\Delta}} \!\!\varphi(z) \;:=\; \int_{\Delta} \Big\{ u(z)\; \text{on the branch}\; u_{\Delta}(z) \Big\}\, \varphi(z).
\ee
We say that $\Delta$ is \emph{loaded} with $u_{\Delta}(z)$. Since on a small neighbourhood around $\Delta$ the form $u_{\Delta}(z) \varphi(z)$ is single-valued, we can apply the ordinary Stokes theorem to find:
\be\label{int-Stokes}
\int_{\partial \Delta} \! u_{\Delta}(z)\, \varphi(z) = \int_{\Delta} d \left(u_{\Delta}(z)\, \varphi(z) \right) = \int_{\Delta} u_{\Delta}(z) \Big( d + \omega \!\w\! \Big) \varphi(z),
\ee
where
\be
\omega := d \log u(z) = \sum_{i=1}^{k} \alpha_i \frac{d f_i(z)}{f_i (z)}
\ee
is a single-valued $1$-form on $\Delta$. The combination in the brackets defines a differential operator $\con := d + \omega \w$, called a \emph{connection}. It is straightforward to check that $\con \cdot \con = 0$, which makes $\con$ an integrable connection \cite{deligne1970equations}. With these definitions, \eqref{int-Stokes} becomes:
\be\label{int-twisted-Stokes}
\int_{\Delta \otimes u_{\Delta}} \!\!\! \con \varphi(z) \;=\; \int_{\partial_\omega \left( \Delta \otimes u_{\Delta} \right)} \! \varphi(z),
\ee
where the remaining part is to specify how the boundary operator $\partial_\omega$ acts on $\Delta \otimes u_\Delta (z)$. Let us illustrate it with a couple of examples. We use the standard notation \cite{nakahara2003geometry} for an $m$-simplex, $\Delta = \la 01 \cdots m \ra$. In the one-dimensional case $\Delta = \la 01 \ra$ we have:
\be
\partial \la 01 \ra = \la 1 \ra - \la 0 \ra \qquad\text{and similarly}\qquad \partial_\omega \left( \la 01 \ra \otimes u_{\la 01 \ra}(z) \right) = \la 1 \ra \otimes u_{\la 1 \ra}(z) - \la 0 \ra \otimes u_{\la 0 \ra}(z).
\ee
Here the branch $u_{\la 1 \ra}(z)$ is induced from $u_{\la 01 \ra}(z)$ at the boundary of $\la 1 \ra$ of $\la 01 \ra$, and similarly for $u_{\la 0 \ra}(z)$. Therefore, the twisted Stokes theorem \eqref{int-twisted-Stokes} in this case becomes:
\be
\int_{\la 01 \ra \otimes u_{\la 01 \ra}} \!\!\! \con \varphi(z) \;=\; \int_{\la 1 \ra \otimes u_{\la 1 \ra}} \!\!\varphi(z) \;-\; \int_{\la 0 \ra \otimes u_{\la 0 \ra}} \!\!\varphi(z),
\ee
where each contribution gives $u(z) \varphi(z)$ evaluated at an appropriate branch at points $z= \la 0\ra$ and $\la 1\ra$.
Similarly, in the two-dimensional case, where $\Delta = \la 012 \ra$ we have:
\be
\partial_\omega \left( \la 012 \ra \otimes u_{\la 012 \ra}(z) \right) = \la 12 \ra \otimes u_{\la 12 \ra}(z) + \la 20 \ra \otimes u_{\la 20 \ra}(z) + \la 01 \ra \otimes u_{\la 01 \ra}(z).
\ee
Here, the twisted cycles associated to the boundaries $\la 12\ra$, $\la 20\ra$, and $\la 01\ra$ are determined by $u_{\la 012\ra}(z)$, which naturally translates to the twisted Stokes theorem:
\be
\int_{\la 012 \ra \otimes u_{\la 012 \ra}} \!\!\! \con \varphi(z) \;=\; \int_{\la 12 \ra \otimes u_{\la 12 \ra}} \!\!\varphi(z) \;+\; \int_{\la 20 \ra \otimes u_{\la 20 \ra}} \!\!\varphi(z) \;+\; \int_{\la 01 \ra \otimes u_{\la 01 \ra}} \!\!\varphi(z).
\ee
A generalization to higher-dimensional simplices is now clear. The twisted boundary operator acts on an $m$-simplex as:
\be\label{boundary-operator}
\partial_\omega \left( \la 01\cdots m \ra \otimes u_{\la 01\cdots m \ra}(z) \right) = \sum_{i=0}^{m} (-1)^{i} \la 01\cdots \hat{i}\cdots m \ra \otimes u_{\la 01\cdots \hat{i}\cdots m  \ra}(z),
\ee
where the hat denotes a removed label. For every triangulable manifold this definition can be used to compute the action of the boundary operator by gluing simplices together.

Let us interpret the above analysis in the language of algebraic topology. In order to track the information about branches we define homology with coefficients in a local system $\mathcal{L}_\omega^\vee$ defined by the differential equation
\be\label{local-system}
\nabla_{\omega}\xi = d\xi + \omega \!\w\! \xi = 0. 
\ee 
It admits a formal solution for $\xi$ of the form $\xi(z) = c / u(z)$, where $c \in \mathbb{C}$ is a constant. The space generated by local solutions of \eqref{local-system} is therefore one-dimensional. Let us cover the manifold $X$ with a locally finite open cover, such that $X=\bigcup_i U_i$, and fix a solution $\xi_i$ on each of the open sets $U_i$. On the intersection of two of them, $U_i$ and $U_j$, we have:
\be
\xi_i(z) = \zeta_{ij}\, \xi_j(z) \qquad\text{for}\qquad z\in U_i \cap U_j,
\ee
where $\zeta_{ij}$ is a constant on $U_i \cap U_j$. Given that a solution $\xi(z)$ on $U_i \cap U_j$ can be expressed as $\xi(z) = \tilde{c}_i \xi_i(z) = \tilde{c}_j \xi_j(z)$ for constants $\tilde{c}_i, \tilde{c}_j \in \mathbb{C}$, we have $\tilde{c}_i = \zeta_{ij}^{-1} \tilde{c}_j$. Therefore, the set of local solutions of \eqref{local-system} defines a flat line bundle, denoted by $\mathcal{L}_\omega^\vee$, obtained by gluing the fibers $\{ \tilde{c}_i\}$ by transition functions $\{ \zeta_{ij}^{-1}\}$. Similarly, we can define a dual line bundle $\mathcal{L}_\omega$, which corresponds to the transition functions $\{\zeta_{ij}\}$. It is generated by local solutions of the differential equation
\be
\nabla_{-\omega}\xi = d\xi - \omega \!\w\! \xi = 0.
\ee
Since the boundary operator \eqref{boundary-operator} coincides with the above system generated by $\mathcal{L}_\omega$, we can define a twisted chain group $C_m(X,\mathcal{L}_\omega)$ with the basis of $\Delta \otimes u_\Delta(z)$. The boundary operator is given by a map:
\be
C_m(X,\mathcal{L}_\omega) \;\xrightarrow{\;\partial_\omega\;}\; C_{m-1}(X,\mathcal{L}_\omega),
\ee
for which one can show $\partial_\omega \circ \partial_\omega = 0$. The definition of the $m$-th twisted de Rham homology group is given by a natural generalization the usual homology group:\footnote{Twisted homology groups $H_{k}(X,\mathcal{L}_\omega)$ with $k<m$ generically vanish \cite{aomoto1975,kita1993,aomoto2011theory}. For the purpose of this work, we will be only interested in top homologies and cohomologies.}
\be\label{def-twisted-homology}
H_{m}(X,\mathcal{L}_\omega) := \ker \partial_\omega \slash\, \text{im}\, \partial_\omega.
\ee
In other words, twisted homology is a space of boundary-less topological cycles with a loading, $\gamma \otimes u_\gamma(z)$, which are not boundaries themselves. We call these elements \emph{twisted} (or \emph{loaded}) cycles.

Let us turn to the associated twisted cohomology, which now has a straightforward definition. Since the function $u(z)$ vanishes at the boundaries of the cycles, the right-hand side of \eqref{int-twisted-Stokes} is equal to zero. This implies that adding a combination $\nabla_{\omega} \xi(z)$ to $\varphi(z)$ does not affect the result of the integration. In other words, $\varphi(z)$ and $\varphi(z) + \nabla_{\omega} \xi(z)$ are in the same cohomology class for any smooth $(m-1)$-form $\xi(z)$. This leads to the definition of the $m$-th twisted cohomology:
\be\label{def-twisted-cohomology}
H^{m}(X,\nabla_{\omega}) := \ker \nabla_\omega \slash\, \text{im}\, \nabla_\omega,
\ee
which means it is a space of cocycles which are closed but not exact with respect to $\nabla_\omega$. We call these elements \emph{twisted} cocycles. In similarity to the twisted homology case, one can also define a dual twisted cohomology with the connection $\nabla_{\omega}^\vee = \nabla_{-\omega}$. We will make use of this fact in the remainder of the paper. We can now use the twisted homology \eqref{def-twisted-homology} and twisted cohomology \eqref{def-twisted-cohomology} to define a non-degenerate pairing:
\be
H_{m}(X,\mathcal{L}_\omega) \times H^{m}(X,\nabla_{\omega}) \;\longrightarrow\; \mathbb{C},
\ee
given by
\be\label{homology-cohomology-pairing}
\la \gamma \otimes u_\gamma,\, \varphi(z) \ra := \int_{\gamma \otimes u_\gamma} \!\!\varphi(z).
\ee
This is a way of formulating the initial integral \eqref{int-multi-valued} in the language of twisted de Rham theory.

Manifolds considered in this work will generically be non-compact. In this case, one ought to consider the \emph{locally finite} twisted homology group $H^{\text{lf}}_m(X,\mathcal{L}_\omega)$ defined using a locally finite cover of $X$. Pairings between different twisted cycles and cocycles require at least one of them to be compact or with compact support \cite{aomoto2011theory}. We will explicitly construct a map from $H_m^{\text{lf}}(X,\mathcal{L}_\omega)$ to the space of compact twisted cycles, $H_m(X,\mathcal{L}_\omega)$, in Section~\ref{subsec-regularization}. We will discuss the use of an inclusion map from $H^{m}(X,\nabla_{\omega})$ to the compactly supported twisted cohomology $H_{c}^{m}(X,\nabla_{\omega})$ in Section~\ref{sec-conclusion}. For mathematically rigorous definitions of these statements see, e.g., \cite{aomoto2011theory}.\footnote{For a treatment of non-compact topological spaces in general, see the textbooks on algebraic topology \cite{hatcher2002algebraic,bott2013differential}.}

\subsection{String Theory Scattering Amplitudes}

\textsc{Much of the structure} of quantum field theories and their generalizations are encapsulated in \emph{scattering amplitudes}. Physically, they calculate the probabilities of given scattering states---such as particles or strings---to interact with each other. Despite the fact that efficient calculation of scattering amplitudes is indispensable in experimentally testing predictions of current models of physics at particle colliders \cite{Olive:2016xmw}, we will be mainly interested in their mathematical structure. Let us focus the discussion on string theory amplitudes.

The very first examples of string amplitudes appeared in the pioneering papers of Veneziano \cite{Veneziano:1968yb}, Virasoro \cite{Virasoro:1969me}, Shapiro \cite{Shapiro:1970gy}, as well as Koba and Nielsen \cite{Koba:1969rw,Koba:1969kh}, long before the formulation of string theory. Since then, calculation of string theory scattering amplitudes developed into a rich field of research of its own, see, e.g., \cite{DHoker:1988pdl,Berkovits:2000fe,Stieberger:2009hq,Mafra:2011nw,Mafra:2011nv}. For historical account of the developments of string theory see \cite{cappelli2012birth}. Here, we will give a brief review of the topics relevant for this paper. Great introduction to the subject is given in the classic textbooks by Green, Schwarz, and Witten \cite{green1988superstring,green1988superstring2}, as well as Polchinski \cite{polchinski1998string,polchinski1998string2}.

Strings come in two types: open and closed. Evolution of strings in spacetime creates a two-dimensional surface called the \emph{worldsheet}. Using the underlying conformal symmetry, we can map the worldsheet into a Riemann surface, which takes the scattering states into vertex operators. Scattering amplitude is then given as an integration of vertex operator correlation function over all their inequivalent positions. The $n$-point open string amplitude takes the form:
\be\label{open-string-amplitude}
\mathcal{A}^{\text{open}}_{\text{full}} = \text{Tr}(T^{\mathrm{a}_1} T^{\mathrm{a}_2} \cdots T^{\mathrm{a}_n}) \int_{\mathfrak{D}(12\cdots n)} \frac{d^n z}{\text{vol SL}(2,\mathbb{R})} \prod_{i<j}\; (z_j - z_i)^{\alpha' s_{ij}}\, F(z) \;+\; \ldots.
\ee
Let us dissect this formula one-by-one. Each string has an associated spacetime momentum $k_i^{\mu}$ for $\mu = 0,1,\ldots,d$, where $d+1$ is the spacetime dimension. We take all momenta to be incoming and impose momentum conservation $\sum_i k_i^\mu = 0$. Each string also has a \emph{colour} $\mathrm{a}_i$ associated to a generator $T^{\mathrm{a}_i}$ of the unitary group $U(N)$. Other possible quantum numbers, such as polarization vectors, are all enclosed in the rational function $F(z)$. In fact, $F(z)$ is a part of the correlation function of vertex operators which depends on the type of string theory used. The multi-valued function $\prod_{i<j}(z_j - z_i)^{\alpha' s_{ij}}$, called the \emph{Koba--Nielsen} factor \cite{Koba:1969rw}, is common to all types of strings. In the exponent we used the parameter $\alpha'$, which is proportional to the inverse of the string tension and serves as a coupling constant of the string amplitude \eqref{open-string-amplitude}. Here $s_{ij} = k_i \cdot k_j$ is an inner product of the momenta known as the \emph{Mandelstam invariant} \cite{Mandelstam:1958xc}. The integration variables $\{ z_1, z_2, \ldots, z_n \}$ are the positions of the vertex operators, which we associate to marked points---or \emph{punctures}---on the boundary of a genus-zero Riemann surface. Due to the inherit $\text{SL}(2,\mathbb{R})$ redundancy of the correlator, one needs to quotient out the action of this group, which is denoted by division by $\text{vol SL}(2,\mathbb{R})$. In practice, it boils down to fixing positions of three punctures, which by convention is taken to be $(z_1, z_{n-1}, z_n) = (0,1,\infty)$. In doing so, one picks up a constant factor $(z_1-z_{n-1})(z_{n-1}-z_n)(z_n - z_1)$ due to the Faddeev--Popov Jacobian \cite{Faddeev:1967fc}. The \emph{disk ordering} $\mathfrak{D}(12\cdots n)$ denotes a region of integration given by $\{z_1 < z_2 < \ldots < z_n\}$ after gauge-fixing. It comes with the associated trace $\text{Tr}(T^{\text{a}_1} T^{\text{a}_2} \cdots T^{\text{a}_n})$ of Chan--Paton factors \cite{Paton:1969je} due to the colour structure of the strings. The ellipsis in \eqref{open-string-amplitude} denote a sum over all $(n-1)!$ cyclically-inequivalent permutations of the vertex operators, each decorated with a trace factor.

It is important to mention that in \eqref{open-string-amplitude} we have only displayed contributions from the genus-zero Riemann surface. In order to obtain the full string theory amplitude, one sums over all possible genera of Riemann surfaces. Genus-zero terms correspond to the \emph{tree-level}---or classical---scattering amplitudes, while the genus-one and higher terms give quantum corrections. For the purpose of this work we will restrict ourselves to tree-level amplitudes only.

The amplitude \eqref{open-string-amplitude} admits a natural splitting into \emph{partial} (or \emph{colour-ordered}) amplitudes $\mathcal{A}^{\text{open}}(\beta)$ defined as coefficients of a Chan--Paton trace with the permutation $\beta$. These will be the objects of our interest. We have:
\be\label{open-string-partial-amplitude}
\mathcal{A}^{\text{open}}(\beta) := \int_{\mathfrak{D}(\beta)} \frac{d^n z}{\text{vol SL}(2,\mathbb{R})} \prod_{i<j}\; \left(z_{\beta(j)} - z_{\beta(i)}\right)^{\alpha' s_{\beta(i), \beta(j)}}\, F(z),
\ee
where the only information about the permuation $\beta$ comes from the disk ordering $\mathfrak{D}(\beta)$ and the choice of the branch for the Koba--Nielsen factor. Scattering amplitudes of closed strings are defined similarly. For their $n$-point scattering we have:
\be\label{closed-string-amplitude}
\mathcal{A}^{\text{closed}} := \int \frac{d^{2n} z}{\text{vol SL}(2,\mathbb{C})} \prod_{i<j}\; |z_{i} - z_{j}|^{2\alpha' s_{ij}}\; F(z)\, F(\bar{z}).
\ee
Here, the integration proceeds over the full moduli space of a genus-zero Riemann surface with $n$ punctures, $\mathcal{M}_{0,n}$. The $\text{SL}(2,\mathbb{C})$ redundancy is fixed by choosing positions of three punctures. The integrand of \eqref{closed-string-amplitude} factors into two functions, a holomorphic and an anti-holomorphic one, which once again depend on the type of string theory under consideration. The common piece is given by the Koba--Nielsen factor. Note that in both \eqref{open-string-partial-amplitude} and \eqref{closed-string-amplitude} we have omitted coupling constants that give rise to an overall normalization factor, see, e.g., \cite{green1988superstring}.

The precise form of the integrands of \eqref{open-string-partial-amplitude} and \eqref{closed-string-amplitude} will not be important for our purposes and can be found, for instance, in \cite{green1988superstring}. It was shown by Mafra, Schlotterer, and Stieberger \cite{Mafra:2011nv,Mafra:2011nw} that open string partial amplitudes can be expanded in a basis of the so-called \emph{Z-theory} amplitudes \cite{Mafra:2016mcc} as follows:
\be\label{open-string-expansion}
\mathcal{A}^{\text{open}}(\beta) = \sum_{\gamma \in \mathcal{C}} n(\gamma)\, Z_{\beta}(\gamma),
\ee
where
\be\label{Z-theory-amplitude}
Z_{\beta}(\gamma) := \int_{\mathfrak{D}(\beta)} \frac{d^n z}{\text{vol SL}(2,\mathbb{R})} \; \frac{\prod_{i<j} \left(z_{\beta(j)} - z_{\beta(i)}\right)^{\alpha' s_{\beta(i), \beta(j)}}}{\left(z_{\gamma(1)}-z_{\gamma(2)}\right)\left(z_{\gamma(2)}-z_{\gamma(3)}\right)\cdots\left(z_{\gamma(n)}-z_{\gamma(1)}\right)}.
\ee
The coefficients of the expansion, $n(\gamma)$, are only a function of kinematic invariants, polarization vectors, and possibly Grassmann variables in the supersymmetric case. The entire dependence on the string parameter $\alpha'$ and the colour ordering $\beta$ is encapsulated in the Z-theory amplitude \eqref{Z-theory-amplitude}. The sum is over a set $\mathcal{C}$ of $(n-3)!$ permutations.\footnote{In fact, \eqref{open-string-expansion} is another instance of a field theory KLT relation \cite{Mafra:2011nv,Mafra:2011nw}. This fact, however, will not play any role in this work.} Such a set is called a Bern--Carrasco--Johansson (BCJ) basis \cite{Bern:2008qj} originally found for Yang--Mills amplitudes and later generalized to the open string ones by Stieberger \cite{Stieberger:2009hq}. The Z-integral \eqref{Z-theory-amplitude} depends on two permutations, $\beta$ serving as a disk ordering, and $\gamma$ which determines the form of the integrand function. Because all the string theoretic properties of open string amplitudes are determined by the Z-theory amplitudes, it will be sufficient to study the integrals \eqref{Z-theory-amplitude} as the primary ingredients in our work. A similar decomposition can be performed in the closed string case \eqref{closed-string-amplitude}. It reads:
\be\label{closed-string-expansion}
\mathcal{A}^{\text{closed}} = \!\!\! \sum_{\beta \in \mathcal{B},\, \gamma \in \mathcal{C}}\!\!\! n(\beta)\, n(\gamma)\, J(\beta | \gamma),
\ee\vspace{-10pt}
where
\be\label{J-intergral}
J(\beta | \gamma) := \int \frac{d^{2n} z}{\text{vol SL}(2,\mathbb{C})} \frac{\prod_{i<j}\; |z_{i} - z_{j}|^{2\alpha' s_{ij}}}{\left(z_{\beta(1)}-z_{\beta(2)}\right)\cdots\left(z_{\beta(n)}-z_{\beta(1)}\right)\left(\overbar{z}_{\gamma(1)} - \overbar{z}_{\gamma(2)}\right) \cdots \left(\overbar{z}_{\gamma(n)} - \overbar{z}_{\gamma(1)}\right)}.
\ee
The sum in \eqref{closed-string-expansion} proceeds over two sets of permutations $\mathcal{B}$ and $\mathcal{C}$, each of length $(n-3)!$. The coefficients $n(\gamma)$ are the same as in the open string case \eqref{open-string-expansion}. The object \eqref{J-intergral} is an integral over the moduli space $\mathcal{M}_{0,n}$ \cite{Stieberger:2014hba}, with the integrand composed of two pieces called the \emph{Parke--Taylor} factors\footnote{The name comes due to the resemblence to the scattering amplitude of gluons with MHV helicity configuration found by Parke and Taylor \cite{Parke:1986gb}.} which also appear in \eqref{Z-theory-amplitude}. In contrast with \eqref{Z-theory-amplitude}, however, \eqref{J-intergral} is symmetric under the exchange of the two permutations $\beta$ and $\gamma$.

Calculation of the above string integrals is a difficult problem that has been approached in many different ways, see, e.g. \cite{Mafra:2011nv,Mafra:2011nw,Broedel:2013aza,Yuan:2014gva,Mafra:2016mcc}. The general approach is to perform an expansion around $\alpha' = 0$. In particular, it is known that in the $\alpha' \to 0$ limit, the integrals $Z_{\beta}(\gamma)$, $J(\beta | \gamma)$, as well as entries of the inverse of the \emph{string theory} KLT kernel $m_{\alpha'}(\beta | \gamma)$ all approach the same answer, up to global scaling:
\be
\lim_{\alpha' \to 0} Z_{\beta}(\gamma) = \lim_{\alpha' \to 0} J(\beta|\gamma) = \lim_{\alpha' \to 0} m_{\alpha'}(\beta|\gamma) = {\alpha'}^{3-n}\, m(\beta | \gamma).
\ee
Here $m(\beta | \gamma)$ are double-partial amplitudes of the so-called \emph{bi-adjoint scalar} \cite{BjerrumBohr:2012mg,Cachazo:2013iea}. This amplitude is given by a sum over all trivalent Feynman diagrams $\mathcal{T}$ which are planar with respect to both $\beta$ and $\gamma$:
\be
m(\beta | \gamma) := (-1)^{w(\beta | \gamma)+1} \!\!\!\sum_{\mathcal{T} \in \mathcal{G}_\beta \cap \mathcal{G}_\gamma} \frac{1}{\prod_{e \in \mathcal{T}} s_e},
\ee
where $\mathcal{G}_\beta$ denotes the space of all trivalent Feynman diagrams planar with respect to the ordering $\beta$, and $e \in \mathcal{T}$ means the set of internal edges of a given diagram $\mathcal{T}$. The Mandelstam invariant $s_e$ equals $p_e^2/2$, where $p_e$ is the momentum flowing through the edge $e$. We have included a sign factor \cite{Mizera:2016jhj} featuring the relative winding number between the two permutations, $w(\beta|\gamma)$. The amplitudes $m(\beta|\gamma)$ are the entries of the inverse of the \emph{field theory} KLT kernel matrix. 

\pagebreak
\section{\label{sec-klt-as-twisted-period-relations}Kawai--Lewellen--Tye Relations as Twisted Period Relations}

\textsc{Twisted de Rham theory} developed primarily in Japan towards the end of twentieth century has been motivated by trying to understand properties of hypergeometric functions. In particular, an interest lies in finding algebraic relations between different hypergeometric functions. The simplest instance of such an identity is a quadratic relation between Euler beta functions, $B(a,b)$:
\be
B(a,b) B(-a,-b) = -\pi \left(\frac{1}{a} + \frac{1}{b}\right) \left( \frac{1}{\tan \pi a} + \frac{1}{\tan \pi b} \right), \qquad\text{where}\qquad B(a,b) := \int_0^1 z^{a-1} (1-z)^{b-1} dz.
\ee
In pursuit of generalizing this relation to other integrals of multi-valued functions, Cho and Matsumoto discovered identities called the \emph{twisted period relations} \cite{cho1995}. In this section we discuss how to apply these relations to the case of string theory scattering amplitudes and show their equivalence with the Kawai--Lewellen--Tye relations \cite{Kawai:1985xq}.

We first review the statement of twisted period relations. Let us consider a twisted homology $H_m(X,\mathcal{L}_\omega)$ and the associated twisted cohomology $H^m(X,\nabla_{\omega})$ on an $m$-dimensional manifold $X$ and choose a basis of twisted cycles $\gamma_i \otimes u_{\gamma_i}(z)$ and twisted cocycles $\varphi_j(z)$ with $i,j = 1,2,\ldots, d$. Recall that due to a twisted version of de Rham theorem, dimensions of both spaces are equal \cite{aomoto2011theory}, i.e., $d := \dim H_m(X,\mathcal{L}_\omega) = \dim H^m(X,\nabla_{\omega})$. We can organize the bilinears between the bases of twisted cycles and cocycles into a $d \times d$ matrix with elements: 
\be\label{period-matrix}
\mathbf{P}_{ij} := \la \gamma_i \otimes u_{\gamma_i},\, \varphi_j(z) \ra = \int_{\gamma_i \otimes u_{\gamma_i}} \varphi_j(z).
\ee
This defines a \emph{twisted period matrix} $\mathbf{P}$.\footnote{Recall that a \emph{period} is an integral of an algebraic function over a domain specified by polynomial inequalities \cite{Kontsevich2001}. \emph{Twisted period} is a natural extension of this definition \cite{aomoto2011theory}.} Similarly, we can choose the dual twisted homology $H_m(X,\mathcal{L}_\omega^\vee)$ together with its associated twisted cohomology $H^m(X,\nabla_{\omega}^\vee)$ on the same manifold. Recall that that \emph{dual} here means that the homology is defined with a multi-valued function $u^{-1}(z)$ insted of $u(z)$. Once again, we choose bases of twisted cycles $\gamma_i^\vee \otimes u_{\gamma_i^\vee}^{-1}(z)$, as well as twisted cocycles $\varphi_j^\vee(z)$ with $i,j = 1,2,\ldots, d$ of the same dimension $d$ as above. This leads to the definition of a dual twisted period matrix $\mathbf{P}^\vee$ with elements:
\be\label{dual-period-matrix}
\mathbf{P}_{ij}^\vee := \la \gamma_i^\vee \otimes u_{\gamma_i^\vee}^{-1},\, \varphi_j^\vee(z) \ra = \int_{\gamma_i^\vee \otimes u_{\gamma_i^\vee}^{-1}} \varphi^\vee_j(z).
\ee
Relating these two matrices requires a definition of additional pairings between twisted homology and cohomology groups. It turns out one can define a non-degenerate pairing
\be\label{intersection-pairing}
H_m(X,\mathcal{L}_\omega) \times H_m(X,\mathcal{L}_\omega^\vee) \;\longrightarrow\; \mathbb{C},
\ee
called the \emph{intersection number of twisted cycles} \cite{MANA:MANA19941660122}. In \eqref{intersection-pairing} at least one of the twisted cycles ought to be compact. It owes its name to the fact that evaluation of this pairing requires the knowledge of how twisted cycles intersect one another topologically, aided with an information of the branch structure of both twisted cycles. Intersection theory of twisted cycles was originally developed by Kita and Yoshida \cite{MANA:MANA19941660122,MANA:MANA19941680111}. We will give precise definition of \eqref{intersection-pairing} in Section~\ref{sec-inverse-klt}, together with the discussion of how to construct a regularization map from $H^{\text{lf}}_m(X,\mathcal{L}_\omega)$ to $H_m(X,\mathcal{L}_\omega)$. For the time being, let us define a $d \times d$ matrix $\mathbf{H}$ built out of the pairings \eqref{intersection-pairing}:
\be
\mathbf{H}_{ij} = \la \gamma_i \otimes u_{\gamma_i},\, \gamma_j^\vee \otimes u_{\gamma_j^\vee}^{-1} \ra.
\ee
Similarly, there exists a pairing between the two twisted cohomologies,
\be\label{cocycle-intersection-pairing}
H^m(X,\nabla_{\omega}) \times H^m(X,\nabla_{\omega}^\vee) \;\longrightarrow\; \mathbb{C},
\ee
known as the \emph{intersection number of twisted cocycles} \cite{cho1995}. In \eqref{cocycle-intersection-pairing} at least one of the twisted cocycles needs to be with compact support. Different ways of evaluating this pairing were given by Deligne and Mostow \cite{zbMATH03996010}, Cho and Matsumoto \cite{cho1995,cho-private-note,matsumoto1998,Matsumoto1998-2}, as well as Ohara \cite{Ohara98intersectionnumbers}. We can now define another $d \times d$ matrix $\mathbf{C}$ with elements:
\be\label{intersection-of-cocycles-matrix}
\mathbf{C}_{ij} = \la \varphi_i(z),\, \varphi_j^\vee(z) \ra.
\ee
Cho and Matsumoto showed \cite{cho1995} that the matrices defined above can be related by:  
\be\label{twisted-period-relations}
\mathbf{C} = \mathbf{P}^\intercal (\mathbf{H}^{-1})^\intercal \mathbf{P}^\vee \qquad\text{or equivalently}\qquad \mathbf{H} = \mathbf{P} (\mathbf{C}^{-1})^\intercal (\mathbf{P}^\vee)^\intercal
\ee
These are the twisted Riemann period relations.\footnote{The name comes due to the resemblance of \eqref{twisted-period-relations} to the standard period relations on Riemann surfaces, see, e.g., \cite{griffiths2014principles,farkas2012riemann}.} As long as the matrices $\mathbf{P},\mathbf{P}^\vee,\mathbf{H},\mathbf{C}$ are defined by bases of their respective homologies and cohomologies, they are invertible. By $\mathbf{P}^\intercal$ we denote a transpose of the matrix $\mathbf{P}$. The relations \eqref{twisted-period-relations} hold under the condition that the cocycles in the bases $\varphi_i(z)$ and $\varphi_j^\vee(z)$ are logarithmic \cite{cho1995}.

Note that the dual twisted homology and cohomology are defined with a multi-valued function $u^{-1}(z)$. In order to apply the above relations to string theory amplitudes, we need to consider a different set of spaces defined with a complex conjugate function $\overbar{u(z)}$ instead. Such a setting was first considered by Hanamura and Yoshida \cite{hanamura1999}, and later studied in the context of Selberg-type integrals by Mimachi and Yoshida \cite{Mimachi2003,Mimachi2004}, see also \cite{mimachi2004-2}. Indeed, a canonical isomorphism $\mathcal{L}_{-\omega} \cong \mathcal{L}_{\overbar{\omega}}$ can be defined when the exponents $\alpha_i$ in $u(z)$ are real and sufficiently generic. From now on we will implicitly use such an isomorphism and work with the dual twisted homology defined by the system $\mathcal{L}_\omega^\vee = \mathcal{L}_{\overbar{\omega}}$ and a dual twisted cohomology defined with the connection $\nabla_{\omega}^\vee = \nabla_{\overbar{\omega}}$. See \cite{Mimachi2003} for details of this construction. The pairing \eqref{intersection-of-cocycles-matrix} then takes the form:\footnote{The name \emph{intersection number} of twisted cocycles is justified only in the case of the dual cohomology defined with $\nabla_{\omega}^\vee = \nabla_{-\omega}$, where the pairing receives contributions only from certain regions of the moduli space. We discuss it in Section~\ref{sec-conclusion}. In the case $\nabla_{\omega}^\vee = \nabla_{\overbar{\omega}}$ there is nothing to \emph{intersect}. To author's best knowledge, the intersection form of cohomology groups \eqref{cohomology-intersection-form} is poorly understood beyond the one-dimensional case.}
\be\label{cohomology-intersection-form}
\la \varphi_i(z),\, \varphi_j^\vee(z) \ra := \int_X |u(z)|^2\; \varphi_i(z) \!\w\! \overbar{\varphi_j^\vee(z)},
\ee
such that the integral converges. Study of the Hodge structure of such integrals was initiated in \cite{hanamura1999}. Let us now turn to the problem of formulating tree-level string theory amplitudes in the language of twisted de Rham theory.

\subsection{Twisted Cycles for String Amplitudes}

\textsc{Open string scattering amplitudes} are defined on the moduli space of genus-zero Riemann surfaces with $n$ punctures, $X=\mathcal{M}_{0,n}$. After gauge fixing the positions of three of them to $(z_1, z_{n-1}, z_n) = (0,1,\infty)$, the regions of integration of the open string amplitudes are given by a disk ordering $\mathfrak{D}(\beta)$ with a permutation $\beta$, which is an $(n-3)$-simplex labelled by $\beta$:
\be\label{simplex}
\Delta_{n-3}(\beta) := \overbar{\{ 0 < z_{\beta(2)} < z_{\beta(3)} < \cdots < z_{\beta(n-2)} < 1 \}}.
\ee
It is embedded in the real section of the moduli space, $ \mathcal{M}_{0,n}(\mathbb{R})$. Here the overbar denotes a closure of the space. Twisted cycles are then defined as follows.
\begin{definition}
	Twisted cycle on $X = \M_{0,n}$ labelled by a permutation $\beta$ is given by 
	\be\label{string-cycles}
	\C(\beta) := \Delta^{o}_{n-3}(\beta) \otimes\, {\mathsf{SL}_\beta}[u(z)],
	\ee
	whose topological part is the interior of the simplex $\Delta_{n-3}(\beta)$. The branch of $u(z)$ for a given twisted cycle is chosen according to the so-called standard loading, denoted by $\mathsf{SL}$. We define it as
	\be\label{standard-loading}
	{\mathsf{SL}_\beta}[u(z)] = \prod_{i<j} \left( z_{\beta(j)} - z_{\beta(i)} \right)^{\alpha' s_{\beta(i),\beta(j)}}.
	\ee
	The set $\normalfont\{ \C(\beta) \,|\, \beta \in (1,\mathfrak{S}_{n-3}(2,3,\ldots,n-2),n-1,n)\}$ of cardinality $(n-3)!$ forms a basis of twisted cycles. Here, $\mathfrak{S}_{n-3}$ denotes permutations of a set of $n-3$ labels.
\end{definition}

\noindent
Twisted cycles are elements of $H_{n-3}^{\text{lf}}(X,\mathcal{L}_{\omega})$. The size of the basis is known to be $(n-3)!$ from the study of Selberg integrals by Aomoto, see, e.g., \cite{doi:10.1093/qmath/38.4.385,aomoto2011theory}, as well as the BCJ basis for open string amplitudes \cite{Stieberger:2009hq}, or equivalently size of the KLT matrix \cite{BjerrumBohr:2010hn}. Of course, one can also choose different bases of twisted cycles labelled by different sets of $(n-3)!$ orderings, not necessarily being related by a permutation operator. The multi-valued function $u(z)$ is given by the Koba--Nielsen factor:
\be\label{Koba-Nielsen}
u(z) := \prod_{i<j} \left(z_i - z_j\right)^{\alpha' s_{ij}} = \prod_{i=2}^{n-2} \left(0 - z_i\right)^{\alpha' s_{1i}} \prod_{i=2}^{n-2} \left(z_i - 1\right)^{\alpha' s_{i,n-1}} \!\!\!\!\prod_{2 \leq i < j \leq n-2}\!\! \left(z_i - z_j\right)^{\alpha' s_{ij}}.
\ee
The twist $1$-form $\omega$ then becomes:
\be
\omega = d \log \prod_{i<j} (z_i - z_j)^{\alpha' s_{ij}} = \alpha' \sum_{i<j} s_{ij}\, d \log (z_{i} - z_{j}) = \alpha' \sum_{i=2}^{n-2} \left( \sum_{j \neq i} \frac{s_{ij}}{z_{i} - z_{j}}\right) dz_i = \alpha' \sum_{i=2}^{n-2} E_i\, dz_i,
\ee
where $E_i := \sum_{j \neq i} s_{ij}/(z_i - z_j)$ are the so-called \emph{scattering equations} \cite{Cachazo:2013gna}. The divisor $D$ is defined by the singular locus of $u(z)$, i.e.,
\be\label{divisor}
D := \bigcup_{i=2}^{n-2}\; \{ z_i = 0 \}\; \bigcup_{i=2}^{n-2}\; \{ z_i - 1 = 0 \} \!\!\bigcup_{2 \leq i < j \leq n-2}\!\! \{ z_i - z_j = 0 \}.
\ee
Since $D$ does not belong to the manifold $X$, the objects \eqref{string-cycles} have no boundaries in $X$, that is $\partial\, \C(\beta) = \varnothing$ for any $\beta$. This justifies the use of the name twisted \emph{cycle}.

The above definition is not fully satisfactory, as it contains singular points when more than two punctures coalesce at once. In order to resolve this issue, we consider a Deligne--Mumford--Knudsen compactification \cite{Deligne1969,MathScand11642,MathScand12001,MathScand12002} of $\mathcal{M}_{0,n}$, given by the so-called \emph{minimal blowup}, $\pi^{-1}(\mathcal{M}_{0,n}) = \widetilde{\mathcal{M}}_{0,n}$ \cite{DeConcini1995,Mimachi2004}. A blowup of a simplex \eqref{simplex} is a polytope called the \emph{associahedron} \cite{10.2307/1993608}. We will study this object more closely in Section~\ref{sec-inverse-klt}. An additional regularization from locally finite twisted homology $H_{n-3}^{\text{lf}}(X,\mathcal{L}_{\omega})$ into $H_{n-3}(X,\mathcal{L}_{\omega})$, where the twisted cycles are compact can be constructed by considering Pochhammer contour and its higher-dimensional generalizations, see, e.g., \cite{kita1993,yoshida2013hypergeometric,aomoto2011theory}. We will show how to obtain it in Section~\ref{subsec-regularization}, and how to use it the study of the field theory limit of open string amplitudes in Appendix~\ref{app-field-theory-limit}.

\subsection{\label{subsec-twisted-cocycles}Twisted Cocycles for String Amplitudes}

\textsc{The dimension of the twisted cohomology} group $H^{n-3}(X,\nabla_{\omega})$ is also $(n-3)!$. A convenient basis for this space studied in the string amplitudes literature \cite{Mafra:2011nv,Mafra:2011nw} is given by the so-called \emph{Parke--Taylor} factors:\footnote{Note that in order to be consistent with the literature, we have not permuted the differential form in the numerator. As a consequence, Parke--Taylor factors for different permutations are related by relabelling and an additional change of sign.}
\be
\PT(\beta) = \frac{dz_1 \w dz_{2} \w dz_{3} \w \cdots \w dz_{n-2} \w dz_{n-1} \w dz_n}{(z_{\beta(1)} - z_{\beta(2)})(z_{\beta(2)} - z_{\beta(3)}) \cdots (z_{\beta(n-1)} - z_{\beta(n)})(z_{\beta(n)} - z_{\beta(1)})} \bigg/ \text{vol SL}(2,\mathbb{R}).
\ee
Here one needs to fix the $\text{SL}(2,\mathbb{R})$ redundancy in the same way as for twisted cycles by taking $(z_1, z_{n-1}, z_n) = (0,1,\infty)$ and compensating with a constant Faddeev--Popov factor:
\be
\text{vol SL}(2,\mathbb{R}) = \frac{dz_1 \w dz_{n-1} \w dz_n}{(z_1 - z_{n-1})(z_{n-1} - z_n)(z_n - z_1)},
\ee
This leads to the following definition for twisted cocycles.
\begin{definition}
	Twisted cocycle on $X = \mathcal{M}_{0,n}$ labelled by a permutation $\beta$ is given by
	\be\label{PT-def}
	\PT(\beta) := \frac{dz_{2} \w dz_{3} \w \cdots \w dz_{n-2}}{(0 - z_{\beta(2)})(z_{\beta(2)} - z_{\beta(3)}) \cdots (z_{\beta(n-2)} - 1)}.
	\ee
	The set $\normalfont\{ \PT(\beta) \,|\, \beta \in (1,\mathfrak{S}_{n-3}(2,3,\ldots,n-2),n-1,n)\}$ of cardinality $(n-3)!$ forms a basis of twisted cocycles.
\end{definition}

\noindent
Twisted cocycles are elements of $H^{n-3}(X,\nabla_{\omega})$. Once again, it is often necessary to consider a blowup of \eqref{PT-def} defined on $\widetilde{\mathcal{M}}_{0,n}$. We illustrate how to perform it in practice in the Appendix~\ref{app-field-theory-limit}, see also \cite{Ohara98intersectionnumbers}. In order to satisfy the assumptions of the twisted period relations \eqref{twisted-period-relations}, it is required that the twisted cycles \eqref{PT-def} are logarithmic. In the following we prove by construction that \eqref{PT-def} is a logarithmic differential form.
\begin{claim}
	The Parke--Taylor factor \eqref{PT-def} can be represented as a logarithmic $(n-3)$-form:
	\be
	\PT(\beta) = (-1)^{n} \sgn(\beta)\, d\log \left(\frac{0 - z_{\beta(2)}}{z_{\beta(2),\beta(3)}}\right) \,\wedge\, d\log \left(\frac{z_{\beta(2),\beta(3)}}{z_{\beta(3),\beta(4)}}\right) \,\wedge\, \cdots \,\wedge\, d\log \left(\frac{z_{\beta(n-3),\beta(n-2)}}{z_{\beta(n-2)} - 1}\right),\nn
	\ee
	where $z_{ab} := z_a - z_b$. Note that the prefactor is a constant.
\end{claim}
\begin{proof}
	We will prove the claim inductively in $n$. For clarity of notation let us specialize to the canonical permutation $\I_n = (12 \cdots n)$ without loss of generality. The cases $n=3,4$ can be checked explicitly:
	\be
	\PT(\I_3) = -1,
	\ee
	\be
	\PT(\I_4) = d \log \frac{0 - z_2}{z_2 - 1} = \left( \frac{1}{z_2} - \frac{1}{z_2 - 1} \right) dz_2 = \frac{dz_2}{(0-z_2)(z_2-1)}.
	\ee
	Assuming that the statement is true for $n-2$ and $n-1$, for $n \geq 5$ we find:
	\begin{align}
	\PT(\I_n) &= (-1)^{n}\; d\log \frac{0 - z_{2}}{z_{23}} \w \cdots \w d\log \frac{z_{n-4,n-3}}{z_{n-3,n-2}} \w d\log \frac{z_{n-3,n-2}}{z_{n-2} - 1}\tr
	\label{thm1-n}&= (-1)^{n} \left( d\log \frac{0 - z_{2}}{z_{23}} \w \cdots \w d\log \frac{z_{n-4,n-3}}{z_{n-3,n-2}} \right) \!\w\! \left( \frac{dz_{n-3}}{z_{n-3,n-2}} - \frac{(z_{n-3}-1)\, dz_{n-2}}{z_{n-3,n-2}\, (z_{n-2}-1)}  \right)\!.\qquad\,
	\end{align}
	The term in the first pair of brackets is almost proportional to $\PT(\I_{n-1})$, however it includes an additional variable $z_{n-2}$ which is a constant in the definition of $\PT(\I_{n-1})$. We need to take it into account by adding an extra differential with respect to $z_{n-2}$. Then the term in the first brackets becomes:
	\be\label{thm1-n-1}
	\frac{\PT(\I_{n-1})}{(-1)^{n-1}} + \left( d\log \frac{0-z_{2}}{z_{23}} \w \cdots \w d\log \frac{z_{n-5,n-4}}{z_{n-4,n-3}} \right) \w \frac{dz_{n-2}}{z_{n-3,n-2}}.
	\ee
	Once again, the term in the brackets is proportional to the lower-point case $\PT(\I_{n-2})$ plus terms including the $1$-form $dz_{n-3}$. However, the additional terms give rise to the $2$-form $dz_{n-3} \w dz_{n-2}$ in the expression \eqref{thm1-n-1}, and hence vanish in the full expression for $\PT(\I_n)$, since they are wedged with the second bracket in \eqref{thm1-n}. Up to these terms \eqref{thm1-n-1} equals
	\be
	\frac{\PT(\I_{n-1})}{(-1)^{n-1}}\,  + \frac{\PT(\I_{n-2})}{(-1)^{n-2}} \w \frac{dz_{n-2}}{z_{n-3,n-2}}.
	\ee
	Note that in both $\PT(\I_{n-1})$ and $\PT(\I_{n-2})$ we have set the punctures fixed at $1$ to have an arbitrary position, denoted by $z_{n-3}$ and $z_{n-4}$ respectively. Plugging the above expression into \eqref{thm1-n} we find:
	\begin{align}
	\PT(\I_n) &= (-1)^{n} \left( \frac{\PT(\I_{n-1})}{(-1)^{n-1}}\,  + \frac{\PT(\I_{n-2})}{(-1)^{n-2}} \w \frac{dz_{n-2}}{z_{n-3,n-2}} \right) \w \left( \frac{dz_{n-3}}{z_{n-3,n-2}} - \frac{(z_{n-3}-1)\, dz_{n-2}}{z_{n-3,n-2}\, (z_{n-2}-1)}  \right)\tr
	&= \frac{1}{z_{n-3,n-2}} \left( \frac{z_{n-3}-1}{z_{n-2}-1}\, \PT(\I_{n-1}) \w dz_{n-2} - \frac{1}{z_{n-3,n-2}}\, \PT(\I_{n-2}) \w dz_{n-3} \w dz_{n-2} \right).
	\end{align}
	We now use the inductive assumption to obtain:
	\begin{align}
	\PT(\I_n) &= \frac{1}{z_{n-3} - z_{n-2}} \bigg( \frac{z_{n-3}-1}{z_{n-2}-1}\, \frac{1}{(0-z_2) (z_2 - z_3)\cdots (z_{n-4}-z_{n-3}) (z_{n-3} - z_{n-2})} \tr
	&\qquad\qquad\qquad\qquad- \frac{1}{z_{n-3} - z_{n-2}}\, \frac{1}{(0-z_2) (z_2 - z_3)\cdots (z_{n-4} - z_{n-3})} \bigg)\, dz_2 \w \cdots \w dz_{n-2}\tr
	&= \frac{dz_2 \w dz_3 \w \cdots \w dz_{n-2}}{(0-z_2)(z_2 - z_3)\cdots (z_{n-3} - z_{n-2})(z_{n-2} - 1)} \left( \frac{z_{n-3} - 1}{z_{n-3} - z_{n-2}} - \frac{z_{n-2} - 1}{z_{n-3} - z_{n-2}} \right)\tr
	&= \frac{dz_2 \w dz_3 \w \cdots \w dz_{n-2}}{(0-z_2)(z_2 - z_3)\cdots (z_{n-3} - z_{n-2})(z_{n-2} - 1)},
	\end{align}
	which completes the proof.
\end{proof}
Note how due to its recursive nature, $\PT(\I_n)$ in its logarithmic form contains Fibonacci number of terms, $F_{n-2}$ \cite{A000045}, that all collapse to a single one \eqref{PT-def} once summed over. One may wonder if generalizations of the Parke--Taylor factor used to describe multi-trace amplitudes \cite{Schlotterer:2016cxa} or general scalar theories \cite{Baadsgaard:2016fel} are logarithmic. In the following we show that they are not, and therefore cannot enter the bases of the twisted homologies used in twisted period relations \eqref{twisted-period-relations}.
\begin{claim}
	The multi-trace Parke--Taylor factors of the form
	\be\normalfont
	\PT(\beta | \gamma | \cdots) = \frac{dz_1 \w dz_2 \w \cdots \w dz_n}{(z_{\beta(1)} - z_{\beta(2)}) \cdots (z_{\beta(|\beta|)} -z_{\beta(1)}) (z_{\gamma(1)} - z_{\gamma(2)}) \cdots (z_{\gamma(|\gamma|)} - z_{\gamma(1)}) \cdots\;} \bigg/ \text{vol SL}(2,\mathbb{R}),
	\ee
	where the permutations $\beta, \gamma, \ldots$ are a partition of $(12\cdots n)$, are not logarithmic on $\widetilde{\M}_{0,n}$.
\end{claim}
\begin{proof}
	Recall that a differential form is logarithmic if it has no higher-order poles along the divisor of $X$ given by the singular locus of $u(z)$. We will show that for the multi-trace Parke--Taylor factor, there always exists a higher-order pole, and therefore it cannot be logarithmic. Let us focus on the subpermutation $\beta$, which without loss of generality we can choose to be $\beta = (12\cdots m)$ for some $2 \leq m \leq n-2$. Since $|\beta| < n-1$, we can fix two of the punctures in the remaining permutations to be $1$ and $\infty$. Let us also take $z_1 = 0$. We then perform a blowup along the face $\{z_1 = z_2 = \cdots = z_m\}$ by taking
	\be
	z_i = \tau y_i \qquad \text{for}\qquad i=1,2,\ldots,m,
	\ee
	as well as set $y_1 = 0$, $y_m = 1$. Changing the variables of integration from $\{z_2, z_3, \ldots, z_m\}$ to $\{\tau,y_2,y_3,\ldots,y_{m-1}\}$, the differentials in the numerator scale as
	\be
	dz_2 \w dz_3 \w \ldots \w dz_m \;\sim\; \tau^{m-2} d\tau,
	\ee
	while the denominator scales as
	\be
	(0 - z_2)(z_2 - z_3) \cdots (z_m - 0) \;\sim\; \tau^{m}.
	\ee
	Hence the contour given by $\{|\tau| = \varepsilon\}$ receives the contribution proportional to $d\tau / \tau^2$, which is a double pole. We conclude that multi-trace Parke--Taylor factors are not logarithmic.
\end{proof}

\subsection{KLT Relations Revisited}

\textsc{With the definitions} of twisted cycles \eqref{string-cycles} and twisted cocycles \eqref{PT-def} we can study their pairings. We also have analogous definitions for the dual spaces $H_{n-3}(X,\mathcal{L}_{\omega}^\vee)$ and $H^{n-3}(X,\nabla_{\omega}^\vee)$, whose bases we label with $\C(\beta)^\vee$ and $\PT(\beta)^\vee$ respectively. Elements of the period matrices \eqref{period-matrix} and \eqref{dual-period-matrix} then become:
\be\label{pairing-form}
\la \C(\beta),\, \PT(\gamma) \ra = Z_{\beta}(\gamma) \qquad\text{and}\qquad \la \C(\beta)^\vee,\, \PT(\gamma)^\vee \ra = Z_{\beta}(\gamma),
\ee
Both of these bilinears give the Z-integrals as defined in \eqref{Z-theory-amplitude}. Similarly, a pairing between two twisted cocycles is given by the J-integral \eqref{J-intergral}:
\be\label{intersection-form}
\la \PT(\beta),\, \PT(\gamma)^\vee \ra = J(\beta | \gamma).
\ee
In Section~\ref{sec-inverse-klt} we will prove that:
\be\label{klt-cycle-cycle}
\la \C(\beta),\, \C(\gamma)^\vee \ra = \left(\frac{i}{2}\right)^{n-3} m_{\alpha'}(\beta | \gamma),
\ee
for an appropriately defined pairing between the two twisted cycles computed by their intersection number. Here, $m_{\alpha'}(\beta | \gamma)$ denotes the $\alpha'$-corrected bi-adjoint scalar amplitudes introduced in \cite{Mizera:2016jhj}. We can build matrices out of the above pairings and apply the twisted period relation \eqref{twisted-period-relations} in order to obtain the relation:\footnote{By $m_{\alpha'}^{-1}(\beta | \gamma)$ we denote the inverse of a matrix $m_{\alpha'}(\gamma | \beta)$ with rows labelled by $\gamma \in \mathcal{C}$ and columns labelled by $\beta \in \mathcal{B}$.}
\be\label{twisted-period-relations-Z}
J(\delta | \epsilon) = \sum_{\beta \in \mathcal{B},\, \gamma \in \mathcal{C}} Z_{\beta}(\delta)\, m_{\alpha'}^{-1}(\beta | \gamma)\, Z_{\gamma}(\epsilon).
\ee
Note that here we have absorbed the constant factor $(i/2)^{n-3}$ from \eqref{klt-cycle-cycle} into the definition a coupling constant of $J(\delta | \epsilon)$. Using the fact that open string amplitudes can be expanded in the basis of Z-integrals, and closed string amplitudes can be expanded in the basis of J-integrals as:
\be\label{expansion}
\mathcal{A}^{\text{open}}(\beta) = \sum_{\delta \in \mathcal{D}} n(\delta)\, Z_{\beta}(\delta) \qquad\text{and}\qquad \mathcal{A}^{\text{closed}} = \!\!\sum_{\delta \in \mathcal{D},\, \epsilon \in \mathcal{E}}\!\! n(\delta)\, n(\epsilon)\, J(\delta | \epsilon), 
\ee
we find
\be\label{KLT-relation}
\mathcal{A}^{\text{closed}} = \!\!\sum_{\beta \in \mathcal{B},\, \gamma \in \mathcal{C}}\!\! \mathcal{A}^{\text{open}}(\beta)\, m_{\alpha'}^{-1}(\beta | \gamma)\, \mathcal{A}^{\text{open}}(\gamma).
\ee
These are the Kawai--Lewellen--Tye relations \cite{Kawai:1985xq}. We conclude that twisted period relations for string theory amplitudes are equivalent to KLT relations. In fact, similar identities can be written for any other string-like models having BCJ representations of the form \eqref{expansion}.

\begin{example}
	Let us illustrate \eqref{KLT-relation} with an example for $n=4$. The sizes of the bases are $1$, and we can choose them to be $\mathcal{B}=\mathcal{C}=\{(1234)\}$. The KLT relations \eqref{KLT-relation} then read:
	\be\normalfont\label{KLT-4}
	\mathcal{A}^{\text{closed}}_4 = \mathcal{A}^{\text{open}}(1234)\, \left(\frac{1}{\tan \pi \alpha' s} + \frac{1}{\tan \pi \alpha' t} \right)^{-1} \!\! \mathcal{A}^{\text{open}}(1234),
	\ee
	where we used the notation $s = s_{12}$ and $t = s_{23}$. We will give a method of calculating the above coefficient of KLT expansion in Section~\ref{subsec-four-point}. The four-point open string amplitude is given by the Veneziano amplitude \cite{Veneziano:1968yb} proportional to the beta function, $B(\alpha' s, \alpha' t)$. Plugging it into \eqref{KLT-4} and using trigonometric identities, we find:
	\be\normalfont
	\mathcal{A}^{\text{closed}}_4 = \frac{\sin \pi \alpha' s\, \sin \pi \alpha' t}{\sin \pi \alpha'(s+t)} B(\alpha' s, \alpha' t)^2 = -\pi {\alpha'}^2 u^2 \frac{\Gamma(\alpha' s)\, \Gamma(\alpha' t)\, \Gamma(\alpha' u)}{\Gamma(1-\alpha' s)\,\Gamma(1-\alpha' t)\,\Gamma(1-\alpha' u)},
	\ee
	where $u = s_{13} = -s - t$ by momentum conservation. This expression is indeed proportional to the Virasoro--Shapiro amplitude for four-point closed string scattering \cite{Virasoro:1969me,Shapiro:1970gy}.
\end{example}

\begin{example}
	For $n=5$ the size of the basis is $2$. Let us take $\mathcal{B}=\{(12345),(12435)\}$ and $\mathcal{C}=\{(13254),(14253)\}$. The KLT relations \eqref{KLT-relation} become:
	\begin{align}
	\mathcal{A}^{\mathrm{closed}}_5 &=
	{\renewcommand{\arraystretch}{1.6}
	\left[ \begin{array}{c}
	\mathcal{A}^{\mathrm{open}}(12345)\\
	\mathcal{A}^{\mathrm{open}}(12435)
	\end{array} \right]^\intercal}
	\left[ \begin{array}{cc}
	\dfrac{1}{\sin \pi \alpha' s_{23} \, \sin \pi \alpha' s_{45}} & 0 \\
	0 & \dfrac{1}{\sin \pi \alpha' s_{24}\, \sin \pi \alpha' s_{35}}
	\end{array} \right]^{-1}
	{\renewcommand{\arraystretch}{1.6}
	\left[ \begin{array}{cc}
	\mathcal{A}^{\mathrm{open}}(13254)\\
	\mathcal{A}^{\mathrm{open}}(14253)
	\end{array} \right]}\tr\tr
	&= \sin \pi \alpha' s_{23}\, \sin \pi \alpha' s_{45} \,\mathcal{A}^{\mathrm{open}}(12345)\, \mathcal{A}^{\mathrm{open}}(13254) \;+\; (3 \leftrightarrow 4).
	\end{align}
	In Section~\ref{subsec-five-point} we will discuss how to calculate entries of the above inverse of the KLT kernel. For more examples of KLT relations we refer the reader to \cite{Mizera:2016jhj}.
\end{example}

\subsection{\label{subsec-circuit-matrix}Basis Expansion and the Circuit Matrix}

\textsc{A related question} in twisted de Rham theory is how monodromy group acts on the period matrix \eqref{period-matrix}, see, e.g. \cite{yoshida2013hypergeometric}. More precisely, the problem translates to finding a matrix $\mathbf{M}$ which relates two period matrices $\mathbf{P}'$ and $\mathbf{P}$ with different choices of bases for twisted cycles, say $\mathcal{D}$ and $\mathcal{D}'$:
\be
\mathbf{P}' = \mathbf{M}\mathbf{P}.
\ee
The representative of the monodromy group, $\mathbf{M}$, is called a \emph{circuit matrix} \cite{yoshida2013hypergeometric}. Following the derivation given in \cite{Mizera:2016jhj}, we can show how to construct the entries of the matrix $\mathbf{M}$ from intersection numbers of twisted cycles in the following way. Let $\mathbf{H}$ be a $d \times d$ matrix of intersection numbers of twisted cycles defined with bases $\mathcal{C}$ and $\mathcal{D}$ for the rows and columns respectively, and ${\mathbf{H}}_k'$ be an $(n-3)!$ vector of intersection numbers of a given twisted cycle labelled by $k$ and the basis $\mathcal{C}$. Similarly, let $\mathbf{P}_l$ be a vector of pairings between the basis $\mathcal{D}$ and a given twisted cocycle labelled by $l$, and $\mathbf{P}_{kl}$ be a pairing between the twisted cycle $k$ and twisted cocycle $l$. We can organize these objects into a $(d+1) \times (d+1)$ matrix, whose determinant equals \cite{bernstein2009matrix}:
\be
\setlength{\arraycolsep}{3pt}
\det \left[ \begin{array}{c|c}
	\mathbf{H} & \mathbf{H}_k' \\
	\hline
	\mathbf{P}_l^\intercal & \mathbf{P}_{kl}' \\
\end{array} \right] = \bigg( \mathbf{P}_{kl}' - (\mathbf{H}_k')^\intercal (\mathbf{H}^{-1})^\intercal \mathbf{P}_l \bigg) \det \mathbf{H} = 0.
\ee
This expression vanishes because the final column of the matrix is linearly dependent of the remaining $d$ columns, which form a basis. Since $\det \mathbf{H}$ is non-vanishing, the term in the brackets ought to be equal to zero. Repeating the same procedure $d \times d$ times, we can build a new period matrix $\mathbf{P}'$, whose basis of twisted cycles is $\mathcal{D}'$, while $\mathbf{P}$ has a basis $\mathcal{D}$. Note that both matrices have the same bases of twisted cocycles. Rows of $\mathbf{H}'$ are labelled by the basis $\mathcal{C}$, while its columns are in $\mathcal{D}'$. The final expression reads:
\be\label{circuit-matrix}
\mathbf{P}' = (\mathbf{H}')^\intercal (\mathbf{H}^{-1})^\intercal \mathbf{P} \qquad\text{implying}\qquad \mathbf{M} = (\mathbf{H}')^\intercal (\mathbf{H}^{-1})^\intercal,
\ee
which gives an explicit realization of the circuit matrix $\mathbf{M}$. Note that this expression is independent of the choice of the basis $\mathcal{C}$, whose labels are contracted in the expression for $\mathbf{M}$. When $\mathcal{D}=\mathcal{D}'$ we have $\mathbf{M}=\I$, as expected. It would be interesting to understand how \eqref{circuit-matrix} arises directly from twisted de Rham theory.

In the case of string amplitudes, this expression translates to:
\be\label{basis-expansion}
\mathcal{A}^{\text{open}}(\beta) = \sum_{\gamma \in \mathcal{C},\, \delta \in \mathcal{D}} m_{\alpha'}(\beta | \gamma)\, m_{\alpha'}^{-1}(\gamma | \delta)\, \mathcal{A}^{\text{open}}(\delta),
\ee
which gives a way of expressing a given open string partial amplitude in a BCJ basis \cite{Bern:2008qj,Stieberger:2009hq} given by a set $\mathcal{D}$ of $(n-3)!$ partial amplitudes $\mathcal{A}^{\text{open}}(\delta)$ for $\delta \in \mathcal{D}$. For examples of how to evaluate \eqref{basis-expansion} see \cite{Mizera:2016jhj}.

\pagebreak
\section{\label{sec-inverse-klt}Inverse KLT Kernel as Intersection Numbers of Twisted Cycles}

\textsc{Intersection theory for twisted cycles} was introduced by Kita and Yoshida in 1992 \cite{MANA:MANA19941660122}, who later developed it further in a series of papers \cite{MANA:MANA19941680111,MANA:MANA173,MANA:MANA200310105}. Since then, intersection numbers have been evaluated for a large family of different types of hypergeometric functions \cite{matsumoto1994,Ohara98evaluationof,Matsumoto1998-2,Yoshiaki-GOTO2015203,majima2000,Cho1995-2,goto2015,goto_matsumoto_2015,doi:10.1142/S0129167X13500948,OST2003}, including Selberg-type integrals \cite{Mimachi2003,Mimachi2004,mimachi2004-2,Ohara98intersectionnumbers,MIMACHI2003209,doi:10.1142/S0218216511008887}. For our purposes, intersection numbers of twisted cycles play a central role in the KLT relations by computing entries of the inverse of the KLT kernel. It is therefore important to understand how to evaluate them in the setting of string intergrals. In this section we discuss a combinatorial way for computing intersection numbers of twisted cycles and prove its equivalence to the diagrammatic rules for calculating $m_{\alpha'}(\beta | \gamma)$ given in \cite{Mizera:2016jhj}.

Let us first review the key aspects of the intersection numbers of twisted cycles. Let $H^{\text{lf}}_m(X,\mathcal{L}_\omega)$ be the $m$-th locally finite twisted homology group on a non-compact $m$-dimensional manifold $X = \mathbb{C}^m \setminus D$, where the divisor $D$ is the singular locus of a multi-valued function $u(z) = \prod_{i=1}^k f_i(z)^{\alpha_i}$. The twist $1$-form $\omega = d \log u(z)$ defines an integrable connection $\nabla_{\omega} = d + \omega \w$. The twisted homology has coefficients in $\mathcal{L}_\omega$, the local system of solutions to the differential equation $d\xi = \omega \w \xi$. Twisted cycles are then elements of $H^{\text{lf}}_m(X,\mathcal{L}_\omega)$. Working under the assumption that the exponents $\alpha_i \in \mathbb{R} \setminus \mathbb{Z}$ of $u(z)$ are sufficiently generic, one can define an isomorphism
\be\label{reg-definition}
H_m^{\text{lf}}(X,\mathcal{L}_\omega) \;\xrightarrow{\;\text{reg}\;}\; H_m(X,\mathcal{L}_\omega),
\ee
which is the inverse of the natural map from $H_m(X,\mathcal{L}_\omega)$ to $H_m^{\text{lf}}(X,\mathcal{L}_\omega)$. We refer to the map \eqref{reg-definition} as \emph{regularization} \cite{aomoto2011theory}. We will give plenty of explicit examples of regularized twisted cycles in the following sections.

Similarly, we have a dual $m$-th locally finite twisted homology group $H_m^{\text{lf}}(X,\mathcal{L}_\omega^\vee)$ with the coefficients in the local system $\mathcal{L}_\omega^\vee$ defined with $d \xi = - \omega \w \xi$. Kita and Yoshida showed \cite{MANA:MANA19941660122} that there exists a non-degenerate pairing,
\be\label{int-pairing}
H_{m}(X,\mathcal{L}_\omega) \times H_{m}^{\text{lf}}(X,\mathcal{L}_\omega^\vee) \;\xrightarrow{\;\;\bullet\;\;}\; \mathbb{C},
\ee
known as the intersection form. Together with the regularization map \eqref{reg-definition}, it defines the \emph{intersection number} of two twisted cycles,
\be
\mathsf{C} = \gamma \otimes u_\gamma(z) \in H_{m}^{\text{lf}}(X,\L) \qquad\text{and}\qquad \mathsf{C}^\vee = \gamma^\vee \otimes u^{-1}_{\gamma^\vee}(z) \in H_{m}^{\text{lf}}(X,\mathcal{L}_{-\omega})
\ee
as
\be\label{def-intersection-number}
\text{reg}\,\mathsf{C} \bullet \mathsf{C}^\vee = \sum_{z \in \gamma \,\cap\, \gamma^\vee} \mathsf{Int}_{z}(\gamma, \gamma^\vee)\, u_\gamma(z)\, u_{\gamma^\vee}^{-1}(z).
\ee
Here, $\mathsf{Int}_{z}(\gamma, \gamma^\vee)$ is the topological intersection number of two topological cycles $\gamma$ and $\gamma^\vee$ at point $z$. The sum proceeds over all intersections between the two cycles. When they intersect non-tangentially---which will be the case throughout this work---the topological intersection number $\mathsf{Int}$ is equal to $+1$ or $-1$ depending on their relative orientation, as follows:
\be\label{topological-int}
\raisebox{-.35\height}{\includegraphics{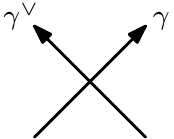}} \;=\; +1 \quad\qquad\text{or}\quad\qquad \raisebox{-.35\height}{\includegraphics{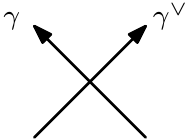}} =\; -1.
\ee

Despite the fact that most of the literature on intersection numbers of twisted cycles has been focused on studying the pairing with the dual homology defined with $\mathcal{L}^\vee = \mathcal{L}_{-\omega}$, one can also apply these ideas to the case complex conjugate case $\mathcal{L}^\vee = \mathcal{L}_{\overbar{\omega}}$ which is more relevant to physics. Hanamura and Yoshida \cite{hanamura1999} considered an isomorphism $\mathcal{L}_{-\omega} \cong \mathcal{L}_{\overbar{\omega}}$ which can be canonically defined if all $\alpha_{i}$ are real and sufficiently generic. Then, for two twisted cycles given by
\be
\mathsf{C} = \gamma \otimes u_\gamma(z) \in H_{m}^{\text{lf}}(X,\L) \qquad\text{and}\qquad \mathsf{C}^\vee = \gamma^\vee \otimes \overline{u_{\gamma^\vee}(z)} \in H_{m}^{\text{lf}}(X,\mathcal{L}_{\overbar{\omega}})
\ee
the intersection number is defined as:
\be\label{def-intersection-number-complex}
\reg\,\mathsf{C} \bullet \mathsf{C}^\vee = \sum_{z \in \gamma \,\cap\, \gamma^\vee} \mathsf{Int}_{z}(\gamma, \gamma^\vee)\, u_{\gamma}(z)\, \overline{u_{\gamma^\vee}(z)} \,/\, |u(z)|^2,
\ee
which is analogous to \eqref{def-intersection-number}. Indeed, when the exponents $\alpha_i$ are real, both definitions agree with each other. For this reason, for considerations of intersection numbers of twisted cycles it will not be important to make distinction between the two cases $\mathcal{L}_{-\omega}$ and $\mathcal{L}_{\overbar{\omega}}$, and hence we will denote twisted cycles belonging to both twisted homologies with same symbols. We will also not distinguish between $\C$ and $\C^\vee$, as they are given by the same definition \eqref{string-cycles}.

Let us focus on the twisted cycles relevant to open string scattering amplitudes. Recall that the multi-valued function defining the local system $\L$ is given by the Koba--Nielsen factor:
\be\label{u-def}
u(z) = \prod_{i<j} (z_i - z_j)^{\alpha' s_{ij}}.
\ee
Here, the Mandelstam invariants $s_{ij} = k_i \cdot k_j$ in the exponents are chosen in such a way that none of the invariants $s_{ij\ldots} = (k_i + k_j + \ldots)^2/2$ is an integer. In the following we will set $\alpha' = 1$ for clarity of notation. The manifold $X$ is the moduli space of genus-zero Riemann surfaces with $n$ punctures, $\mathcal{M}_{0,n}$. Twisted cycles $\C(\beta)$ on this space were defined in \eqref{string-cycles} with the standard loading operator $\mathsf{SL}$, which chooses the branch of the Koba--Nielsen factor for a given permutation $\beta$ in a canonical way. Using this definition, in the case of $n=4$ we have:
\be
\C(1234) = \{ 0 < z_2 < 1 \} \otimes z_2^{s_{12}}(1-z_2)^{s_{23}} = \overrightarrow{(0,1)} \otimes z^s (1-z)^t,
\ee
where we denote the only manifold coordinate as $z = z_2$ and the exponents with the usual notation $s = s_{12}$ and $t = s_{23}$. In the case of $n=5$ the basis has two elements:
\be
\C(12345) = \{ 0 < z_2 < z_3 < 1 \} \otimes z_2^{s_{12}} (1-z_2)^{s_{24}} (z_3 - z_2)^{s_{23}} z_3^{ s_{34}} (1-z_3)^{ s_{34}},
\ee
\be
\C(13245) = \{ 0 < z_3 < z_2 < 1 \} \otimes z_2^{s_{12}} (1-z_2)^{ s_{24}} (z_2 - z_3)^{ s_{23}} z_3^{ s_{34}} (1-z_3)^{s_{34}}.
\ee
One can also define other bases of twisted cycles $\C(\beta)$. They have a straightforward definition analogous to \eqref{string-cycles}. For instance, in the next section we will make us of the four-point twisted cycles:
\be
\C(2134) = \overrightarrow{(-\infty,0)} \otimes (-z)^s (1-z)^t \qquad\text{and}\qquad \C(1324) = \overrightarrow{(1,\infty)} \otimes z^s (z-1)^t.
\ee
Before evaluating intersection numbers let us give an explicit construction of the regularization map \eqref{reg-definition} for twisted cycles $\C(\beta)$, as well as discuss how they are affected by the blowup procedure \cite{DeConcini1995}.

\subsection{\label{subsec-regularization}Regularization of Twisted Cycles}

\textsc{The cycles relevant} for string amplitudes \eqref{string-cycles} are non-compact. Since the definition of the intersection number requires at least one of the twisted cycles to be compact, we need to employ a regularization. In this section we discuss an explicit construction of such a map, based on the Pochhammer contour and its higher-dimensional generalizations, see, e.g., \cite{kita1993,yoshida2013hypergeometric,aomoto2011theory}.

Let us review how the standard Pochhammer contour is constructed. We start by considering the integral:
\be\label{beta-integral}
I := \int_0^1 z^s (1-z)^t\, \varphi(z),
\ee
where $s,t \notin \mathbb{Z}$ and $\varphi(z)$ is any single-valued $1$-form. As defined, the integral converges only for sufficiently positive values of $s$ and $t$. In order to make the it convergent for all values of these parameters, one can employ an alternative contour of integration $\gamma$, known as the Pochhammer contour:
\be
\gamma\, := \includegraphics[valign=c]{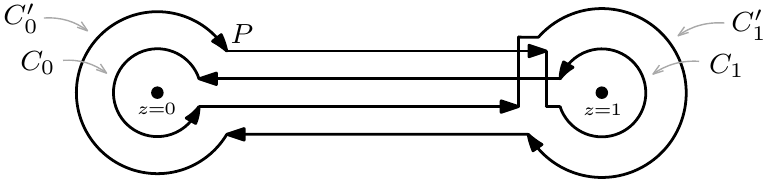}
\ee

\noindent
This contour winds around the two branch points $z=0,1$ once in both directions. We picture the branch cuts as extending from $z=0,1$ downwards to $-i\infty$. Let us track how this contour is related to the one used in \eqref{beta-integral}. Starting from the point $P$ and moving right, we first obtain the contribution equal to $I$. After winding around $z=1$ in a positive direction along $C_1$, one picks up a phase factor $e^{2\pi i t}$, so that the next stretch towards $z=0$ equals to $-e^{2\pi i t} I$, where the minus comes from a different orientation that \eqref{beta-integral}. Next, winding around $z=0$ gives an additional factor of $e^{2\pi i s}$ from $C_0$, so that the following contribution becomes $e^{2\pi i (s+t)} I$. Winding around $z=1$, this time in a negative direction $C_1^\prime$, takes the phase factor back to $e^{2\pi i t}$, so that the final contribution is $e^{2\pi i t} I$. After performing another turn around $z=0$ in a negative direction given by $C_0^\prime$, we land at the point $P$ on the original branch. Summing up the contributions, we have:
\be
\oint_{\gamma} z^s (1-z)^t \,\varphi(z) = \left( 1 - e^{2\pi i t} + e^{2\pi i (s+t)} - e^{2\pi i s} \right) \int_0^1 z^s (1-z)^t \,\varphi(z),
\ee
or equivalently
\be
\int_0^1 z^s (1-z)^t\, \varphi(z) = \oint_{\gamma'} z^s (1-z)^t \,\varphi(z) \qquad\text{with}\qquad \gamma' := \frac{\gamma}{\left( e^{2\pi i s} - 1\right) \left( e^{2\pi i t} - 1\right)}.
\ee
Let us split the contour $\gamma'$ into three parts: regions near the two branch points $z=0,1$, and the interval $\overrightarrow{(\varepsilon, 1-\varepsilon)}$. In order to be precise, we will use a small parameter $\varepsilon$ as the radius of the circular contours. The contributions near the branch point at $z=0$ give:
\be
\frac{C_0 + C_0^\prime}{\left(e^{2\pi i s}-1\right)\left(e^{2\pi i t}-1\right)} = \frac{\left( e^{2\pi i t} - 1 \right) S(\varepsilon,0)}{\left(e^{2\pi i s}-1\right)\left(e^{2\pi i t}-1\right)} = \frac{S(\varepsilon,0)}{e^{2\pi i s}-1},
\ee
where by $S(a,z)$ we denote a positively oriented circular contour with centre at $z$ and starting at a point $a$. Similarly, around $z=1$ we find the contribution
\be
\frac{C_1 + C_1^\prime}{\left(e^{2\pi i s}-1\right)\left(e^{2\pi i t}-1\right)} = \frac{\left(1 - e^{2\pi i s}\right) S(1-\varepsilon, 1)}{\left(e^{2\pi i s}-1\right)\left(e^{2\pi i t}-1\right)} = - \frac{S(1-\varepsilon,1)}{e^{2\pi i t}-1}.
\ee
Finally, the contours along the real axis simply give $\overrightarrow{(\varepsilon, 1-\varepsilon)}$. Putting everything together, we find the regularization of the original cycle $\overrightarrow{(0,1)}$ to be:
\begin{align}
\reg\, \overrightarrow{(0,1)} \;&:=\; \gamma^\prime\tr
&\,=\; \frac{S(\varepsilon,0)}{e^{2\pi i s}-1} + \overrightarrow{(\varepsilon, 1-\varepsilon)} - \frac{S(1-\varepsilon,1)}{e^{2\pi i t}-1}\tr
\label{reg-0-1}&\,=\; \includegraphics[valign=c]{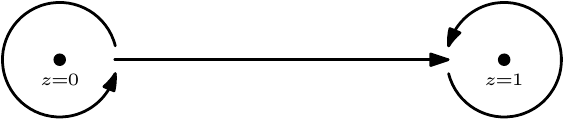}\;.
\end{align}
Here, we have introduced a graphical notation to denote the regularized cycle. It is understood that the circular parts of the contour come multiplied with the additional factors $1/(e^{2\pi i s}-1)$ and $-1/(e^{2\pi i t}-1)$ that are not represented explicitly. We will make a repeated use of this regularization in the following sections. Note that we have been implicitly working with a twisted cycle $\overrightarrow{(0,1)} \otimes z^s (1-z)^t$ relevant for string amplitude calculations.

Generalizations to higher-dimensional cycles can be made by performing a similar regularization \cite{aomoto2011theory}. Since locally we can describe a manifold $X$ as a direct product of lower-dimensional spaces, we can employ the regularization \eqref{reg-0-1} near the singularities on these product spaces. In the case of $X = \mathcal{M}_{0,n}$ with $n \geq 5$, however, there is an additional difficulty coming from the fact that the singular locus of $u(z)$ is not normally crossing. For example, in the case of $n=5$ the function $u(z)$ defining the local system $\L$ is singular at
\be
\{ z_2 = 0\} \cup \{ z_2 - 1 = 0 \} \cup \{ z_2 - z_3 = 0 \} \cup \{ z_3 = 0 \} \cup \{ z_3 - 1 = 0\},
\ee
which has degenerate points at $(z_2, z_3) = (0,0)$, $(1,1)$, and also $(\infty,\infty)$. The way forward is to consider a blowup of this space \cite{Deligne1969,MathScand11642,MathScand12001,MathScand12002,DeConcini1995}, denoted by $\widetilde{\mathcal{M}}_{0,5} = \pi^{-1}(\mathcal{M}_{0,5})$, where all triple singular points get resolved.
\begin{figure}[t]
\centering
\includegraphics[width=\textwidth]{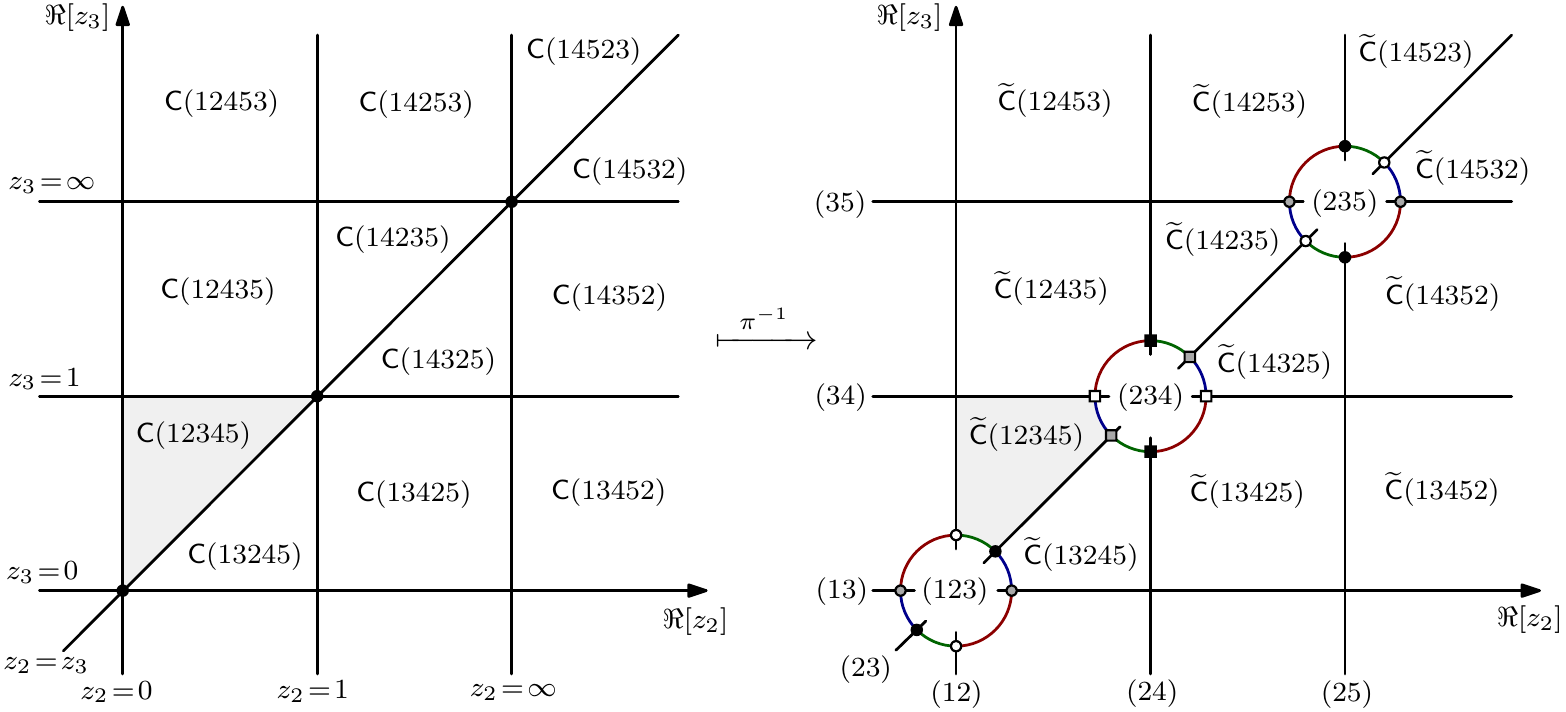}
\caption{\label{fig-blowup}Real slice of the moduli space of Riemann spheres with five punctures, $\mathcal{M}_{0,5}$, and its image under the blowup map $\pi^{-1}(\M_{0,5}) = \widetilde{\M}_{0,5}$.}
\end{figure}
In Figure~\ref{fig-blowup} we have illustrated the real section of $\M_{0,5}$, denoted by $\M_{0,5}(\mathbb{R})$, where the twisted cycles live before the blowup, as well as its image, $\widetilde{\M}_{0,5}(\mathbb{R})$. Note that in this representation we brought the point at infinity to a finite position for convenience. The resulting space is divided into twelve chambers separated by the singular lines. Each of the lines has an associated label corresponding to the exponent of the given zero in $u(z)$, or equivalently a phase factor that one picks up upon crossing the branch line. For example, the line defined by $\{z_2 - 1 = 0\}$ is labelled with $(24)$, since it corresponds to the factor $(z_2 - 1)^{s_{24}}$ in $u(z)$.

Blowup has been performed in the neighbourhood of the points $(0,0)$, $(1,1)$, and $(\infty,\infty)$, resulting in three new locally-defined curves labelled by $(123)$, $(234)$, and $(235)$. For example, near $(z_2,z_3) = (0,0)$ the blowup introduced a line $(123)$ corresponding to the factor ${z_2}^{s_{12}}(z_2 - z_3)^{s_{23}}{z_3}^{s_{13}}$, whose exponents sum up to $s_{12} + s_{23} + s_{13} = s_{123}$. Points labelled with the same symbol on these new curves are identified, and so are the segments between them. Each of the vertices can be uniquely specified as intersection of two lines, for instance the point $(z_2,z_3)=(0,1)$ is written as $(12) \cap (34)$.

Each of the twelve chambers after the blowup forms a polygon known as the \emph{associahedron}, $K_4$.\footnote{Historically, skeleton of the associahedron first appeared the doctoral thesis of Tamari in 1951 \cite{zbMATH03086819}. In 1963, Stasheff gave a realization of the associahedron as a cell complex in his work on associativity of H-spaces \cite{10.2307/1993608,10.2307/1993609}. For this reason, associahedron is often referred to as the \emph{Stasheff polytope}. Since then, many realizations of the polytope have been constructed, see, e.g., \cite{stasheff1997operads,LEE1989551,Loday2004,Hohlweg2007,doi:10.1093/imrn/rnn153,Ceballos2015}. For a historical account see, e.g., the introduction of \cite{Ceballos2015}. Connection to the moduli space $\widetilde{\M}_{0,n}$ was first found by Kapranov in \cite{1992alg.geom.10002K,KAPRANOV1993119}, and from a combinatorial point of view later by Devadoss \cite{Devadoss98tessellationsof,Devadoss_combinatorialequivalence}. It was also independently rediscovered by Yoshida in the context of hypergeometric functions \cite{Yoshida1996}.} Since twisted cycles $\tC(\beta)$ are defined as associahedra with a uniquely specified standard loading given by \eqref{standard-loading}, we will sometimes not distinguish between the two. Notice that the canonical twisted cycle $\tC(12345)$ is neighbouring all the other cycles, except for $\tC(14253)$, by either an edge or a vertex. Adjacency relations between different associahedra will be important in the evaluation of intersection numbers of twisted cycles. Regularized twisted cycles have a natural definition analogous to \eqref{reg-0-1}. For example, $\reg\, \tC(12345)$ in a small neighbourhood of the edge $(23)$ can be represented as:\footnote{Orientations of the cycles are naturally induced from the right-handed manifold $X$. More precisely, a multi-dimensional residue $\{ |g_i(z)| = \varepsilon \}$ is oriented by $\bigwedge_i d \arg g_i > 0$ \cite{yoshida2013hypergeometric}. For the purpose of this section, however, this fact will not be important as we will consider pairings between twisted cycles, for which possible signs due to orientations always cancel out.}
\be\label{Pochhammer-1}
\reg\, \tC(12345) \Big|_{(23)} \;=\; \left\{ \includegraphics[valign=c]{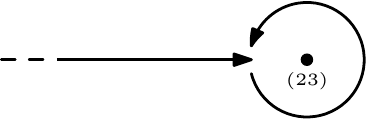} \right\} \w \left\{ \includegraphics[valign=c]{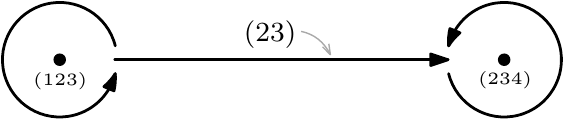}\right\},
\ee
where in both cases the horizontal direction is embedded in the real part of $\widetilde{\M}_{0,5}$ illustrated in Figure~\ref{fig-blowup}. Similarly, near the point $(12) \cap (123)$ we have:
\be\label{Pochhammer-2}
\reg\, \tC(12345) \Big|_{(12) \cap (123)} \;=\; \left\{ \includegraphics[valign=c]{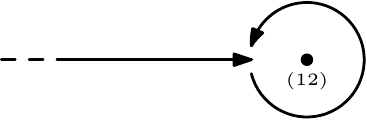} \right\} \w \left\{ \includegraphics[valign=c]{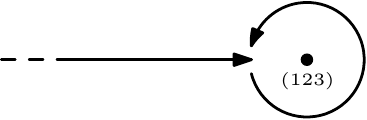} \right\}.
\ee

In general, one considers the Deligne--Mumford--Knudsen compactification \cite{Deligne1969,MathScand11642,MathScand12001,MathScand12002} of $\mathcal{M}_{0,n}$, denoted by $\widetilde{\M}_{0,n}$, in which the singular locus of $u(z)=0$ is normally crossing. It is given by the procedure called the \emph{minimal blowup} \cite{DeConcini1995,Mimachi2004}. It is known that real part of each chamber of $\widetilde{\M}_{0,n}$ is isomorphic to an associahedron $K_{n-1}$, see, e.g., \cite{KAPRANOV1993119,Devadoss98tessellationsof}. We will give properties of associahedra for general $n$ in Section~\ref{general-case}, after studying examples of intersection numbers for $n=4,5$, which will illustrate how they are connected to adjacency relations between different associahedra. Generalized Pochhammer contour for $K_{n-1}$ is defined analogously to \eqref{Pochhammer-1} and \eqref{Pochhammer-2}. We can now give a precise definition of the pairing between twisted cycles, which gives rise to the entries of $\mathbf{H}$.

\begin{definition}
Non-degenerate pairing between two twisted cycles is given by
\be\normalfont
\langle \mathsf{C}(\beta), \mathsf{C}(\gamma) \rangle := \reg\, \widetilde{\C}(\beta) \bullet \widetilde{\C}(\gamma),
\ee
where $\widetilde{\C}(\beta)$ and $\widetilde{\C}(\gamma)$ are two, not necessarily distinct, twisted cycles defined as a blowup of \eqref{string-cycles}. For simplicity we will use the same notation for $n=4$, even though in this case there is no need for a blowup.
\end{definition}

\subsection{\label{subsec-four-point}Four-point Examples}

\textsc{We start evaluation} of intersection numbers with the simplest example of $n=4$, which will illustrate most of the core ideas at play. We first consider the case of the self-intersection number of the twisted cycle $\C(1234)$. In order to avoid degeneracy on the interval $(\varepsilon, 1-\varepsilon)$, let us make a small deformation of one of the cycles into a sine-like curve, on top of the regularization \eqref{reg-0-1} for the other cycle:
\begin{align}
\langle \C(1234), \C(1234) \rangle &= \left( \reg\,\overrightarrow{(0,1)} \otimes z^s (1-z)^t \right) \bullet \left( \overrightarrow{(0,1)}_{\sin} \otimes z^s (1-z)^t \right) \tr
&= \includegraphics[valign=c]{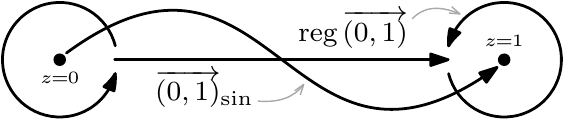}\tr
\label{self-int-4}&= -\frac{1}{e^{2\pi i s}-1} - 1 - \frac{1}{e^{2\pi i t}-1}.
\end{align}
There are three intersection points: near $z=0$, at $z=1/2$, and near $z=1$. The first contribution gives $1/(e^{2\pi i s}-1)$ from the definition \eqref{reg-0-1} times $-1$ arising from the topological intersection number \eqref{topological-int} for this relative orientation of the cycles. Similarly, the second factor is simply $-1$ due to the relative orientation at the intersection point at $z=1/2$. The final factor is $-1/(e^{2\pi i t}-1)$ times $+1$ due to the orientation.

Intersection numbers are independent of the deformation of the second twisted cycle \cite{MANA:MANA19941660122}. For instance, we can calculate it with one of the cycles deformed into a small upside-down sine curve to obtain:
\begin{align}
\langle \C(1234), \C(1234) \rangle &= \left( \reg\,\overrightarrow{(0,1)} \otimes z^s (1-z)^t \right) \bullet \left( \overrightarrow{(0,1)}_{-\!\sin} \otimes z^s (1-z)^t \right) \tr
&= \includegraphics[valign=c]{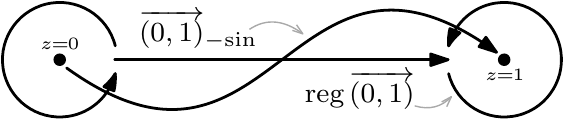}\tr
\label{self-int-4b}&= -\frac{e^{2\pi i s}}{e^{2\pi i s}-1} + 1 - \frac{e^{2\pi i t}}{e^{2\pi i t}-1}.
\end{align}
This time, the two end-point intersection numbers have picked up monodromy factors. Near $z=0$ we have $1/(e^{2\pi i s}-1)$ from the definition of \eqref{reg-0-1} times the phase factor $e^{2\pi i s}$ times $-1$ due to the orientation. Similar reasoning gives the contribution from the neighbourhood of $z=1$. The mid-point intersection has changed orientation and hence give the contribution $+1$. Another choice is a deformation into an arc curve:
\begin{align}
\langle \C(1234), \C(1234) \rangle &= \left( \reg\,\overrightarrow{(0,1)} \otimes z^s (1-z)^t \right) \bullet \left( \overrightarrow{(0,1)}_{\text{arc}} \otimes z^s (1-z)^t \right) \tr
&= \includegraphics[valign=c]{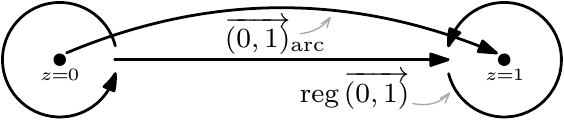}\tr
\label{self-int-4c}&= -\frac{1}{e^{2\pi i s}-1} - \frac{e^{2\pi i t}}{e^{2\pi i t}-1},
\end{align}
which receives contributions from only two intersection points, which we have analyzed before separately. Finally, we have the deformation:
\begin{align}
\langle \C(1234), \C(1234) \rangle &= \left( \reg\,\overrightarrow{(0,1)} \otimes z^s (1-z)^t \right) \bullet \left( \overrightarrow{(0,1)}_{-\!\text{arc}} \otimes z^s (1-z)^t \right) \tr
&= \includegraphics[valign=c]{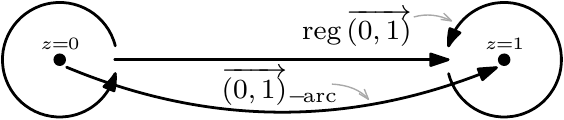}\tr
\label{self-int-4d}&= -\frac{e^{2\pi i s}}{e^{2\pi i s}-1} - \frac{1}{e^{2\pi i t}-1}.
\end{align}
It is straightforward to show that all the above calculations \eqref{self-int-4}, \eqref{self-int-4b}, \eqref{self-int-4c}, and \eqref{self-int-4d} give the same answer:
\be\label{C-1234-1234}
\langle \C(1234), \C(1234) \rangle = \frac{i}{2} \left( \frac{1}{\tan \pi s} + \frac{1}{\tan \pi t}\right).
\ee

Let us now turn to studying intersection numbers of two distinct twisted cycles. Intersecting $\C(1234)$ with $\C(2134)$ one obtains:
\begin{align}
\langle \C(1234), \C(2134) \rangle &= \left( \reg\,\overrightarrow{(0,1)} \otimes z^s (1-z)^t \right) \bullet \left( \overrightarrow{(-\infty,0)} \otimes (-z)^s (1-z)^t \right) \tr
&= \includegraphics[valign=c]{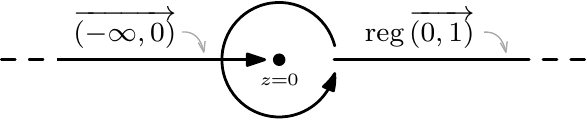}\tr
\label{C-1234-2134}&= \frac{e^{\pi i s}}{e^{2\pi i s} - 1} = \frac{i}{2}\left( -\frac{1}{\sin \pi s} \right).
\end{align}
This time, there is only one intersection point near $z=0$ giving $1/(e^{2\pi i s} - 1)$ times the monodromy factor $e^{\pi i s}$. The topological intersection number gives $+1$. In the remaining case of intersecting twisted cycles $\C(1234)$ and $\C(1324)$ we have:
\begin{align}
\langle \C(1234), \C(1324) \rangle &= \left( \reg\,\overrightarrow{(0,1)} \otimes z^s (1-z)^t \right) \bullet \left( \overrightarrow{(1,\infty)} \otimes (z)^s (z-1)^t \right) \tr
&= \includegraphics[valign=c]{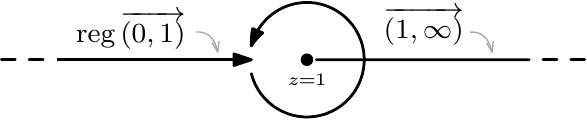}\tr
\label{C-1234-1324}&= \frac{e^{\pi i t}}{e^{2\pi i t} - 1} = \frac{i}{2}\left( -\frac{1}{\sin \pi t} \right),
\end{align}
which comes from the contribution near $z=1$. Note that in both cases \eqref{C-1234-2134} and \eqref{C-1234-1324}, the minus sign in the final answer can be tracked down to the fact that the two cycles involved induce opposite orientation on the boundaries at $z=0$ and $z=1$ respectively. For example, in the case \eqref{C-1234-1324} the boundary of the first cycle $\C(1234)$ is $\partial \overrightarrow{(0,1)} = \{1\} - \{0\}$, and for the second cycle $\C(1324)$ is $\partial \overrightarrow{(1,\infty)} = \{\infty\} - \{1\}$, so the boundaries at $z=1$ contributing to the intersection number have opposite orientations.

All of the remaining combinations of four-point twisted cycles can be obtained by relabelling the cases considered above. As we will see, the four-point cases \eqref{C-1234-2134} and \eqref{C-1234-1324} from this section will also serve as building blocks for intersection numbers for higher multiplicities.

\subsection{\label{subsec-five-point}Five-point Examples}

\textsc{Before moving on} to the most general case, let us study several five-point examples to gain some intuition about higher-dimensional twisted cycles. Without loss of generality we can fix the first twisted cycle to be $\widetilde{\C}(12345)$ and consider its intersections with other cycles. A representation of the real section of the moduli space $\widetilde{\M}_{0,5}$ was given in Figure~\ref{fig-blowup}, from which one can read off the adjacency of different twisted cycles.

We first consider the self-intersection number $\langle \C(12345), \C(12345) \rangle$. Kita and Yoshida showed \cite{MANA:MANA19941680111} that one can define a deformation of the second cycle similar to the sinusoid we used in the $n=4$ case \eqref{self-int-4}. The deformation is made in such a way that the self-intersection number receives contributions from neighbourhoods of the barycenter of the associahedron $K_4$ itself, barycenters of all its facets, and its vertices. Due to the regularization employed, locally near a vertex given by $H_1 \cap H_2$, where $H_1$ and $H_2$ are two facets, we receive a contribution $1/(e^{2\pi i s_{H_1}}-1)(e^{2\pi i s_{H_2}}-1)$, near the barycenter of each facet $H_1$ we get $1/(e^{2\pi i s_{H_1}}-1)$, and near the barycenter of the whole associahedron we get $1$. Explicitly, we have:
\begin{align}
\langle \C(12345), \C(12345) \rangle &= 1 + \frac{1}{e^{2\pi i s_{12}}-1} + \frac{1}{e^{2\pi i s_{23}}-1} + \frac{1}{e^{2\pi i s_{34}}-1} + \frac{1}{e^{2\pi i s_{45}}-1} + \frac{1}{e^{2\pi i s_{51}}-1} \tr
& \hspace{1.8em} + \frac{1}{\left(e^{2\pi i s_{12}}-1\right)\left(e^{2\pi i s_{34}}-1\right)} + \frac{1}{\left(e^{2\pi i s_{23}}-1\right)\left(e^{2\pi i s_{45}}-1\right)}\tr
& \hspace{1.8em} + \frac{1}{\left(e^{2\pi i s_{34}}-1\right)\left(e^{2\pi i s_{51}}-1\right)} + \frac{1}{\left(e^{2\pi i s_{45}}-1\right)\left(e^{2\pi i s_{12}}-1\right)}\tr
\label{self-int-5}& \hspace{1.8em} + \frac{1}{\left(e^{2\pi i s_{51}}-1\right)\left(e^{2\pi i s_{23}}-1\right)},
\end{align}
which is a sum over contributions from five vertices, five facets, and one polygon. See \cite{MANA:MANA19941680111} for details of the derivation. As a matter of fact, \eqref{self-int-5} admits an alternative, more concise form:
\begin{align}
\langle \C(12345), \C(12345) \rangle &= \left( \frac{i}{2} \right)^2 \Bigg( 1 + \frac{1}{\tan \pi s_{12}\, \tan \pi s_{34}} + \frac{1}{\tan \pi s_{23}\, \tan \pi s_{45}} + \frac{1}{\tan \pi s_{34}\, \tan \pi s_{51}}\tr
&\hspace{14.3em}+ \frac{1}{\tan \pi s_{45}\, \tan \pi s_{12}} + \frac{1}{\tan \pi s_{51}\, \tan \pi s_{23}} \Bigg).
\end{align}

Other cases can be obtained by reducing to previously calculated results. For example, two twisted cycles $\widetilde{\C}(12345)$ and $\widetilde{\C}(13245)$ share the facet $(23)$. Working locally in its neighbourhood, we can write the intersection number as a product of the one in the real direction orthogonal to $(23)$ times the boundary $(23)$ itself:
\begin{align}\label{C-12345-13245}
\langle \C(12345), \C(13245) \rangle &= \includegraphics[valign=c]{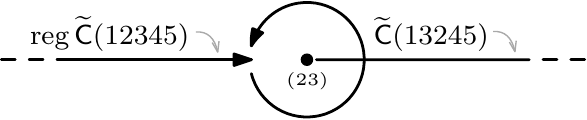} \tr
&= \frac{e^{\pi i s_{23}}}{e^{2\pi i s_{23}} - 1} \left( \reg\, \widetilde{\C}(12345) \Bigg|_{(23)} \!\!\!\bullet\; \widetilde{\C}(13245) \Bigg|_{(23)} \right)\tr
&= -\frac{i}{2} \frac{1}{\sin \pi s_{23}} \left( \includegraphics[valign=c]{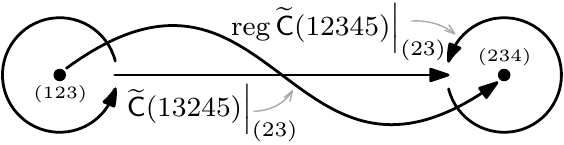} \right)\tr
&= - \left(\frac{i}{2}\right)^2 \frac{1}{\sin \pi s_{23}} \left( \frac{1}{\tan \pi s_{45}} + \frac{1}{\tan \pi s_{51}} \right).\end{align}
Other cases follow the same algorithm. The remaining topology to consider is that of the intersection of $\widetilde{\C}(12345)$ with $\widetilde{\C}(12453)$. This time, these two cycles intersect at a vertex point $(12) \cap (45)$. We first consider the real direction orthogonal to $(12)$ and then the intersection within $(12)$. Being careful about the orientations of the cycles we find:
\begin{align}\label{C-12345-12453}
\langle \C(12345), \C(12453) \rangle &= \includegraphics[valign=c]{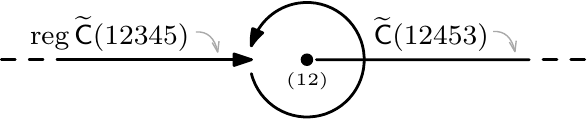} \tr
&= \frac{e^{\pi i s_{12}}}{e^{2\pi i s_{12}} - 1} \left( \reg\, \widetilde{\C}(12345) \Bigg|_{(12)} \!\!\!\bullet\; \widetilde{\C}(12453) \Bigg|_{(12)} \right)\tr
&= -\frac{i}{2} \frac{1}{\sin \pi s_{12}} \left( \includegraphics[valign=c]{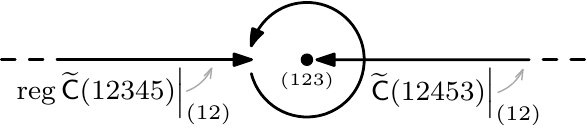} \right)\tr
&= -\left(\frac{i}{2}\right)^2 \frac{1}{\sin \pi s_{12}}\,\frac{1}{\sin \pi s_{45}}.
\end{align}
Here, in the second step, both cycles induce the same orientation on the facet $(123)$, giving an overall plus sign contribution. We will come back to the point of orientations induced on boundaries in the next section.

There is only one chamber in $\widetilde{\M}_{0,5}(\mathbb{R})$ which is not adjacent to $\tC(12345)$, pictured near the top of Figure~\ref{fig-blowup}. It corresponds to the twisted cycle $\tC(13524)$. Because the two cycles do not intersect, we have
\be
\langle \C(12345), \C(13524) \rangle = 0.
\ee
In general, if two chambers are not adjacent, the corresponding intersection number vanishes.

Having studied several examples, the general strategy for evaluating intersection numbers is now clear: after identifying the intersection face $F$ of the two cycles, we project onto facets containing $F$ one by one until the problem reduces to calculating self-intersection numbers for smaller twisted cycles. We will now prove that for general $n$ this procedure reproduces the results of \cite{Mizera:2016jhj} and can be streamlined using simple diagrammatic rules.

\subsection{\label{general-case}Proof of the General Case}

\textsc{Let us review the structure} of $\widetilde{\M}_{0,n}(\mathbb{R})$ in the general case. It is known that this space is divided into $(n-1)!/2$ chambers, each isomorphic to an $(n-3)$-dimensional associahedron $K_{n-1}$, see, e.g., \cite{Devadoss98tessellationsof,Devadoss_combinatorialequivalence}. The space is divided by $2^{n-1} - n - 1$ \cite{A000247} hyperplanes bounding the associahedra, including the ones at infinity. For concreteness, let us specialize to the associahedron defined with the identity permutation, $\I_n := (12\cdots n)$, which we denote by
\be
K_{n-1}(\I_n) := \overbar{\pi^{-1} \{ 0 < z_2 < z_3 < \cdots < z_{n-2} < 1 \}},
\ee
where the overbar means closure of this space, so that bounding facets are also included. Associahedra for other permutations are defined analogously. Twisted cycles on the blowup space $\widetilde{\M}_{0,n}$ are then given as interior of the associahedron loaded with the function $u(z)$ with an appropriate phase given by the standard loading \eqref{standard-loading}:
\be
\tC(\beta) = K^{\mathrm{o}}_{n-1}(\beta) \otimes \mathsf{SL}_\beta [u(z)].
\ee
This is a blowup of the definition \eqref{string-cycles}. Note that in this way we have identified only half of the $(n-1)!$ cyclically-inequivalent permutations. This is because each associahedron comes with an orientation induced from the right-handed space $\widetilde{\M}_{0,n}(\mathbb{R})$. The remaining half of the twisted cycles with reversed permutation $\beta^\intercal$, for example $\I_n^\intercal = (n \cdots 21)$, can be related to the $(n-1)!/2$ set by
\be
\tC(\beta^\intercal) = (-1)^{n}\, \tC(\beta),
\ee
which means that when $n$ is even, twisted cycles corresponding to associahedra with both orientations are identified.\footnote{An alternative is to consider an orientable double cover of $\widetilde{\M}_{0,n}(\mathbb{R})$, whose combinatorics has been studied in \cite{doi:10.1137/130947532,DEVADOSS201575}. Such space also has a known decomposition into $(n-2)!$ \emph{permutohedra} \cite{gelfand2009discriminants}, which in the language of amplitudes corresponds to the choice of a Del Duca--Dixon--Maltoni basis \cite{DelDuca:1999rs}.} In the odd case, the minus sign arises because of the change of integration region and gauge fixing condition for $\{z_1, z_{n-1}, z_n\}$. Note that a given permutation corresponds to a right-handed associahedron if the labels $\{z_1, z_{n-1}, z_n\}$ come in the canonical ordering, and to a left-handed one otherwise.

Following \cite{Mimachi2004}, we will label the $n(n-3)/2$ facets bounding the chamber $K_{n-1}(\I_n)$ with:
\be\label{hyperfaces}
(12\; \cdots\; i),\quad (23 \;\cdots\; i+1),\quad \ldots,\quad (n-i,\; n-i+1,\; \cdots,\; n-1) \qquad \text{for}\qquad i=2,3,\ldots,n-2.\quad
\ee
Each of these facets is a direct product of two other associahedra \cite{Devadoss98tessellationsof,brown2006multiple}. In our notation, for a face labelled by $(\omega)$ we have:
\be\label{hyperface-isomorphism}
(\omega) \;\cong\; K_{|\omega|}(\omega,\, -I) \times K_{n-|\omega|}(\I_n \!\setminus\! \omega,\, I),
\ee
where by $\I_n \!\setminus\! \omega$ we mean the complement of $\omega$ in $\I_n$. We have introduced a new label $I$, which can be thought of as corresponding to a puncture at infinity in both disk orderings. It inherits the phases from the Koba--Nielsen factor $u(z)$ corresponding particles in the set $\omega$, and hence the puncture with the label $I$ can be represented as having associated momentum
\be
k_I^\mu := \sum_{i \in \omega} k_i^\mu = - \sum_{i \in \I_n \!\setminus \omega} k_i^\mu.
\ee
Similarly, in the second associahedron, the label $-I$ has an associated momentum $-k_I^\mu$. Note that when $|\omega|=2$, $K_2$ is a point and hence the corresponding facets can be thought of as being isomorphic to a single associahedron.

Every codimension-$k$ face $F$ of $K_{n-1}(\I_n)$, for $k=1,2,\ldots,n-3$, can be uniquely specified as an intersection of $k$ facets $H_1,H_2,\ldots,H_k$ from the above set \eqref{hyperfaces}, i.e., $F=H_1 \cap H_2 \cap \cdots \cap H_k$. For each face $F$, the condition \emph{disjoint/contained} is satisfied for all pairs of facets $H_i = (ab\cdots)$ and $H_j=(cd\cdots)$. It says that their labels are either disjoint, $ab\cdots \cap cd\cdots = \varnothing$, or one contains the other, i.e., $ab\cdots \subset cd\cdots$ or $cd\cdots \subset ab\cdots$. For example, for the associahedron $K_{4}(\I_5)$ we have five facets:
\be
(12), \quad (23), \quad (34), \quad (123), \quad (234),
\ee
and its five codimension-$2$ faces, or vertices, are given by:
\be
(12) \cap (34), \quad (12) \cap (123), \quad (23) \cap (123), \quad (23) \cap (234), \quad (34) \cap (234),
\ee
which can be read off from the Figure~\ref{fig-blowup}.

It is known that two associahedra sharing a facet $H$ from the family \eqref{hyperfaces} have permutations that differ by a transposition of the labels of $H$ \cite{Yoshida1996}. For example, $K_4(12345)$ and $K_4(14325)$ share the facet $(234)$. Whenever two associahedra are adjacent through a codimension-$k$ face, they can be reached by a series of $k$ such transposition moves, up to a change of orientation. Conversely, if such a series does not exist, then two associahedra are not adjacent. For instance, $K_4(12345)$ and $K_4(12453)$ and adjacent through the vertex $(123) \cap (23)$, which means that one can pass between them by crossing through $(12)$ and $(123)$ in either order. At the same time, $K_4(12345)$ is not adjacent to $K_4(13524)$, since they do not share any facets, see Figure~\ref{fig-blowup}. As a generalization of \eqref{hyperface-isomorphism}, a codimension-$k$ face is isomorphic to a product of $k+1$ associahedra \cite{Devadoss98tessellationsof}.

Finally, let us remark on orientations that associahedra induce on their faces. For each facet $(\omega)$ from the set \eqref{hyperfaces}, $K_{n-1}(\I_n)$ and its neighbour induce the same orientation on $(\omega)$ if $|\omega|$ is odd, and an opposite one if $|\omega|$ is even \cite{Mimachi2004}. For example, $K_4(12345)$ and $K_4(14325)$ induce the same orientation on $(234)$, while $K_4(12345)$ and $K_4(13245)$ induce on opposite orientation on $(23)$, as can be seen from Figure~\ref{fig-blowup}. Orientations of the lower-dimensional faces can be deduced from applying the same rules recursively. For a combinatorial description of the boundary operator acting on associahedra see \cite{Yoshida1996}.

In combinatorics, associahedron $K_{n-1}$ is a convex polytope whose vertices correspond to all legal ways of inserting bracketings in an word of length $n-1$ in the following way. A pair of brackets is assigned to each of the $n(n-3)/2$ facets \eqref{hyperfaces}. Then, two facets meet if and only if their bracketings are compatible, i.e., satisfy the disjoint/contained condition. Repeating this procedure, every codimension-$k$ face $F$ corresponds to a correct insertion of $k$ pairs of brackets. The number of such faces is given by $T(n-2,k+1)$ \cite{A033282}. Another interpretation, originally due to Loday \cite{Loday2004}, is that of rooted trees with $n-1$ leaves, where a face $F$ is a tree with $k+1$ nodes. We illustrate this in Figure~\ref{fig-associahedra}.\footnote{For more visualizations of associahedra and tiling of moduli spaces see the work of Devadoss, e.g., \cite{Devadoss98tessellationsof,Devadoss_combinatorialequivalence,CARR20062155,devadoss2012shape}.}. We will think of the rooted trees as Feynman diagrams \cite{PhysRev.76.769}.
\begin{figure}[t]
	\centering
	\includegraphics{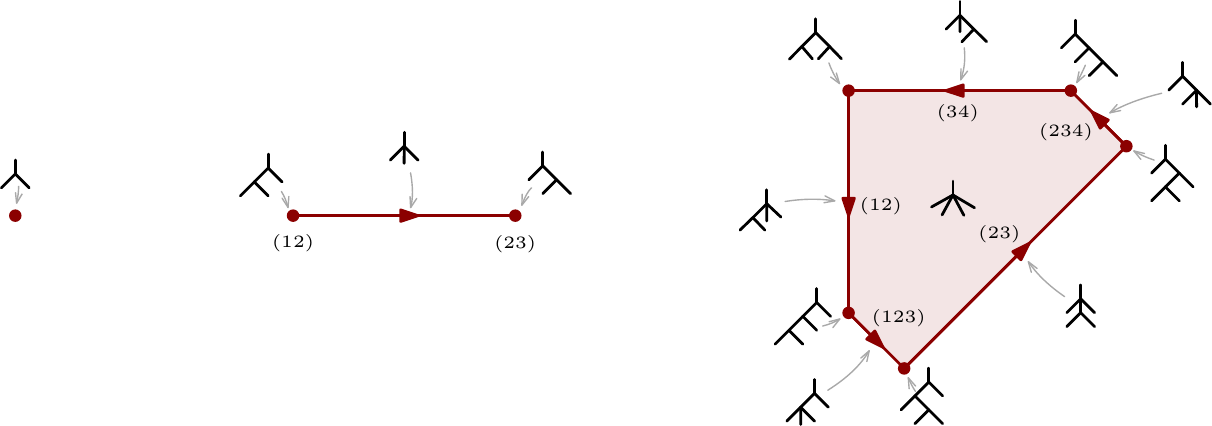}
	\caption{\label{fig-associahedra}Combinatorial interpretation of the associahedra $K_2(\I_3)$, $K_3(\I_4)$, and $K_4(\I_5)$ in terms of rooted trees. Each face has an associated factorization channel.}
\end{figure}

Let us now prove the statement that intersection numbers of twisted cycles give rise to the rules for evaluating the inverse KLT kernel $m_{\alpha'}(\beta|\gamma)$ given in \cite{Mizera:2016jhj}. We split the arguments into two parts. Firstly, we show that self-intersection numbers are proportional to diagonal amplitudes $m_{\alpha'}(\I_n|\I_n)$ in Lemma \ref{lemma-1}. Secondly, we show that the rules for evaluating arbitrary intersection numbers reduce to the self-intersections as building blocks according to the graphical rules of \cite{Mizera:2016jhj} in Theorem \ref{theorem-1}.

\pagebreak
\begin{lemma}\label{lemma-1}
The self-intersection number of the twisted cycle with identity permutation, $\tC(\I_n)$, is equal to the diagonal $\alpha'$-corrected bi-adjoint scalar amplitude $m_{\alpha'}(\I_n | \I_n)$ given in \cite{Mizera:2016jhj} up to a global factor,
\be\label{lemma-1-claim}
\la \C(\I_n), \C(\I_n) \ra = \left(\frac{i}{2}\right)^{n-3} \!\! m_{\alpha'}(\I_n | \I_n).
\ee
\end{lemma}
\begin{proof}
Kita and Yoshida showed \cite{MANA:MANA19941680111} that for general $n$ the self-intersection number is given as a sum over contributions from barycenters of all the codimension-$(0,1,\ldots,n-3)$ faces of the associahedron. The contribution coming from a codimension-$k$ face $F = H_1 \cap H_2 \cap \cdots H_k$ is a product of $k$ terms $1/( e^{2\pi i s_{H_i}} - 1 )$ for every facet $H_i$ intersecting at $F$. More explicitly, we have:
\be\label{self-int-general}
\la \C(\I_n), \C(\I_n) \ra = (-1)^{n-1} \sum_{k=0}^{n-3} \;\sum_{F = H_1 \cap \cdots \cap H_k} \frac{1}{\left(e^{2\pi i s_{H_1}}-1\right) \left(e^{2\pi i s_{H_2}}-1\right) \cdots \left(e^{2\pi i s_{H_k}}-1\right)},
\ee
where we have also included the global sign \cite{MANA:MANA19941680111}. The term in the sum corresponding to $k=0$, i.e., the one coming from the barycenter of the whole associahedron is regarded as $1$. Examples of the evaluation of \eqref{self-int-general} were given in \eqref{self-int-4} for $n=4$, and in \eqref{self-int-5} for $n=5$.

Another way of thinking about the self-intersection number \eqref{self-int-general} is using the interpretation of the associahedron described by Feynman diagrams, as illustrated in Figure~\ref{fig-associahedra}. In this way, the sum \eqref{self-int-general} proceeds over all possible Feynman diagrams in an auxiliary theory described by the following rules. Every internal edge with momentum $p^\mu$ gives rise to a propagator $-1/(e^{i\pi p^2} - 1)$. The theory also has an infinite number of Feynman vertices with valency $p=3,4,5,\ldots$, each coming with a factor of $(-1)^{p-1}$. It is straightforward to check that the signs give rise to the correct prefactor in \eqref{self-int-general}. It is useful to construct equations of motion for such a theory. Let us use a normalization such that factors of $i/2$ from \eqref{lemma-1-claim} are absorbed into propagators and vertices:
\be\label{eom-1}
-\frac{i}{2}\left( e^{i \pi \Box} - 1 \right) \phi = \sum_{p=3}^{\infty} \left(-\frac{2}{i}\right)^{p-3} \phi^{p-1} = \frac{\phi^2}{1 - 2 i \phi}.
\ee
Here, $\phi$ is a real scalar matrix-valued field, and we have denoted $\Box := \partial_\mu \partial^\mu$. The left-hand side gives a normalized propagator $-2/(i(e^{i \pi p^2} - 1))$, while the terms on right-hand side give rise to $p$-valent vertices with factors $(-2/i)^{p-3}$. The additional minus signs are responsible for the factor of $(-1)^{n-3}$ in \eqref{self-int-general}. One can verify that such a change of normalization yields a global factor for an amplitude that counterweights the prefactor of \eqref{lemma-1-claim}.

It was also conjectured in \cite{Mizera:2016jhj} that the $\alpha'$-corrected bi-adjoint scalar amplitudes $m_{\alpha'}(\I_n | \I_n)$ can be expanded using another auxiliary theory with propagators given by $1/\tan \frac{\pi \Box}{2}$. The equation of motion then reads:
\be\label{eom-2}
\left(\tan \frac{\pi \Box}{2}\right) \varphi = V'[\varphi], 
\ee
where the scattering field is denoted by $\varphi$ and the functional $V'[\varphi]$ describes Feynman rules for the vertices. The goal is to prove this theory yields the same amplitudes as the one described by \eqref{eom-1}, with the function $V'[\varphi]$ generating Catalan numbers as proposed in \cite{Mizera:2016jhj}.

Scattering amplitudes can be obtained from both \eqref{eom-1} and \eqref{eom-2} using the following standard procedure, see, e.g., \cite{itzykson2012quantum}. One introduces a coupling constant $g$, such that when $g=0$ the theory becomes free, i.e., the scattering field has no self-interactions. It is then possible to Taylor expand the field around $g=0$, so that equations of motion can be solved iteratively. On the support of this solution, the left-hand sides of \eqref{eom-1} and \eqref{eom-2} become generating functions of integrated scattering amplitudes. It is important that the left-hand sides contain the inverse propagator, which is responsible for striping away the only remaining external propagator. The bottom line is that in order for both equations of motion to produce the same amplitudes, the right-hand sides of both \eqref{eom-1} and \eqref{eom-2} have to be equal on the support of both equations of motion.

Given this knowledge, let us find $V'[\varphi]$ that gives the same amplitudes in both theories. Equating right-hand sides of \eqref{eom-1} and \eqref{eom-2} gives:
\be\label{equating-eom}
\left(\tan \frac{\pi \Box}{2}\right) \varphi = -\frac{i}{2}\left( e^{i \pi \Box} - 1 \right) \phi \qquad \text{and} \qquad V'[\varphi] = \frac{\phi^2}{1 - 2 i \phi}.
\ee
Let us expand the tangent in the first equation to get:
\be
\left( e^{i \pi \Box} - 1 \right) \left( \frac{1}{i}\left( e^{i \pi \Box} + 1 \right)^{-1} \varphi + \frac{i}{2} \phi \right) = 0.
\ee
Since the term in the second brackets is not in the kernel of $\left( e^{i \pi \Box} - 1 \right)$ for generic momenta, it has to vanish. Formally multiplying the term in the second brackets by the operator $i \left( e^{i \pi \Box} + 1 \right)$ from the left, we obtain:
\be
0 = \varphi - \frac{1}{2} \left( e^{i \pi \Box} + 1 \right) \phi = \varphi - \frac{1}{2} \left( e^{i \pi \Box} - 1 \right) \phi - \phi.
\ee
We recognize the second term as being proportional to the left-hand side of the equation of motion \eqref{eom-1}. On the support of \eqref{eom-1} we have:
\be
\varphi - i V'[\varphi] - \phi = 0.
\ee
Finally, using the second equation in \eqref{equating-eom}, we can eliminate $\phi$ in order to get a constraint on the functional $V'[\varphi]$:
\be
V'[\varphi]^2 - V'[\varphi] + \varphi^2 = 0. 
\ee
This equation has two solutions for $V'[\varphi]$, however one of them has a constant independent of $\varphi$, which would not have an interpretation as a Feynman vertex in \eqref{eom-2}. The other solution is
\be
V'[\varphi] = \frac{1}{2} \left( 1 - \sqrt{1 - 4\varphi^2} \right) = \sum_{\substack{p=3\\ \text{odd}}}^{\infty} C_{(p-3)/2}\, \varphi^{p-1}
\ee
which is a generating function for the Catalan numbers $C_{(p-3)/2}$ \cite{A000108}. There is an infinite number of vertices of odd valency $p=3,5,7,9,\ldots$, each contributing to a factor $C_{(p-3)/2}$ equal to $1,1,2,5,\ldots$ respectively. This verifies the conjecture posed in \cite{Mizera:2016jhj} and concludes the proof of \eqref{lemma-1-claim}.
\end{proof}

\begin{remark}
Total number of terms in the $1/(e^{i \pi p^2} - 1)$ representation is given by the Schr\"oder--Hipparchus number \cite{A001003}. Total number of terms in the $1/\tan \frac{\pi \Box}{2}$ representation is given by the series \cite{A049124}. For $n \geq 4$ the latter is always smaller.
\end{remark}

\begin{theorem}\label{theorem-1}
The intersection number of two twisted cycles $\tC(\beta)$ and $\tC(\gamma)$ is equal to the $\alpha'$-corrected bi-adjoint scalar amplitude $m_{\alpha'}(\beta | \gamma)$ given in \cite{Mizera:2016jhj} up to a global factor,
\be\label{theorem-1-claim}
\la \C(\beta), \C(\gamma) \ra = \left(\frac{i}{2}\right)^{n-3} \!\! m_{\alpha'}(\beta | \gamma).
\ee
\end{theorem}
\begin{proof}
We will show that evaluation of intersection numbers gives rise to a recursion relation which is the same as the graphical rules found in \cite{Mizera:2016jhj}.

Let $\tC(\beta)$ and $\tC(\gamma)$ be two twisted cycles on $\widetilde{\M}_{0,n}$. A codimension-$h$ intersection face $F$ of the corresponding associahedra $K_{n-1}(\beta)$ and $K_{n-1}(\gamma)$ can be written as
\be\label{assoc-intersection}
F := K_{n-1}(\beta) \cap K_{n-1}(\gamma) = H_1 \cap H_2 \cap \ldots \cap H_h.
\ee
If $F = \varnothing$ then $\la \C(\beta), \C(\gamma) \ra = 0$, since the cycles are not adjacent. If $F = K_{n-1}(\beta) = K_{n-1}(\gamma)$, then necessarily $\tC(\beta) = \tC(\gamma)$ up to an orientation, which gives the self-intersection number $\pm \la \C(\beta), \C(\beta)\ra$ reducing to the case proven in Lemma \ref{lemma-1}.

Otherwise, let us first fix orientations of both twisted cycles to be the same. If a change of orientation was needed and $n$ is odd then the intersection number picks up a minus sign. Since $F$ is also a codimension-$h$ face belonging to both $K_{n-1}(\beta)$ and $K_{n-1}(\gamma)$, the labels of the facets in the set $\{ H_1, H_2, \ldots, H_h \}$ are necessarily pairwise disjoint/contained. Let us pick one such facet $H$, such that all other $H_i$ either contain $H$ or are disjoint with $H$. Permutation $\beta$ then splits naturally into two parts, $\beta = (H,\beta \!\setminus\! H)$, where $\beta \!\setminus\! H$ denotes the complement of $H$ in $\beta$. Similarly, for $\gamma$ we have $\gamma = (H, \gamma \!\setminus\! H)$.

Let us consider intersection of these two twisted cycles locally as a product of the one in the real orthogonal direction $H^\perp$ times the one within $H$. Since the intersection in $H^\perp$ reduces to the previously studied case \eqref{C-1234-2134}, we have:
\begin{align}
\la \C(\beta), \C(\gamma) \ra &= \frac{i}{2}\frac{(-1)^{|H|-1}}{\sin \pi s_H}\;  \la \C(\beta) \vert_H,\, \C(\gamma)\vert_H \ra \tr
\label{recursion-relation}&= \frac{i}{2} \frac{(-1)^{|H|-1}}{\sin \pi s_H}\;  \left< \C(H,\, -I),\, \C(H,\, -I) \right> \times \left< \C(\beta \!\setminus\! H,\, I),\, \C(\gamma \!\setminus\! H,\, I) \right>,
\end{align}
where we have used that the facet $H$ is a product of two smaller twisted cycles according to \eqref{hyperface-isomorphism}.
The new twisted cycles have loading naturally induced from the one of $\tC(\beta)$ and $\tC(\gamma)$. A potential minus sign arises since $\beta$ and $\gamma$ induce different orientations on $H$ when $|H|$ is even. The new label $I$ corresponds to momentum $k_I^\mu = \sum_{i \in H} k_i^\mu$. Since $|H| \geq 2$, the right-hand side of \eqref{recursion-relation} is a product of a self-intersection number and an intersection number for a cycle with smaller $n$. Thus, it provides a recursion relation that can be used to evaluate an arbitrary intersection number. Simple arithmetic reveals that the overall prefactor becomes $(i/2)^{n-3}$ after performing all the recursive steps.

It follows that the intersection number \eqref{theorem-1-claim} receives contributions from the intersection face $F = H_1 \cap H_2 \cap \cdots \cap H_h$ constructed out of the factors $1/\sin \pi s_{H_i}$ and $1/\tan \pi s_{H'}$ for $H' \subset H_i$.

Let us demonstrate how to conveniently calculate \eqref{recursion-relation} using graphical rules of \cite{Mizera:2016jhj}.\footnote{Similar combinatorial rules for studying adjacency of associahedra in $\widetilde{\M}_{0,n}(\mathbb{R})$ were described in \cite{Devadoss98tessellationsof} using an interpretation of the associahedron $K_{n-1}$ as a configuration space of triangulations of $n$-sided polygons.} We will do so in two steps, first giving the rule for the absolute value of \eqref{theorem-1-claim} and then its sign. One can arrive at the permutation $\gamma$ by a series of $h$ label flips $H_1, H_2, \ldots, H_h$, which are the same as in \eqref{assoc-intersection}, up to a final orientation change. Let us illustrate both permutations as circles connecting the labels $(1,2,\ldots,n)$ in the corresponding orders. We start with two orderings $\beta$, and perform a series of flips $H_1, H_2, \ldots, H_h$ on the second permutation to arrive at the ordering $\gamma$, possibly up to a global orientation change, as follows:
\be
\includegraphics[valign=c]{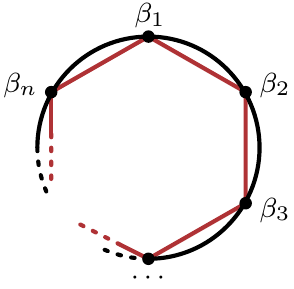} \quad\xmapsto{\;H\;}\quad \includegraphics[valign=c]{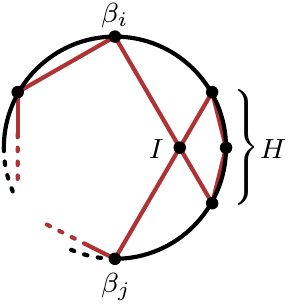} \quad\xmapsto{\{H_1,H_2,\ldots,H_h\} \setminus {H}}\quad \includegraphics[valign=c]{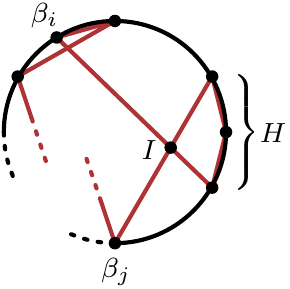}
\ee
Here we have arranged the flips so that $H$ used in \eqref{recursion-relation} is the first one performed. The remaining flips $\{H_1, H_2, \ldots, H_h\} \setminus \{H\}$ do not change the fact that there exists a crossing point $I$ when the labels in $H$ are brought arbitrarily close together. The rule is then to associate a self-intersection number $\left< \C(H,\, -I) , \C(H,\, -I) \right>$ to the polygon created by labels $(H,-I)$, a factor $1/\sin \pi s_H = 1/\sin \pi s_I$ to the intermediate edge, and an intersection number $\left< \C(\beta \!\setminus\! H,\, I) , \C(\gamma \!\setminus\! H,\, I) \right>$ to the remainder of the diagram. Repeating this procedure recursively, one obtains the full intersection number \eqref{theorem-1-claim}.

Finally, we prove that the sign of \eqref{theorem-1-claim} is given by $(-1)^{w(\beta|\gamma)+1}$, where $w(\beta | \gamma)$ is the relative winding number of the two permutations, following the prescription of \cite{Mizera:2016jhj}. Keeping track of signs in the above algorithm, we start with two identical permutations $\beta$, for which $w(\beta|\beta) = 1$ gives a plus sign, as expected. Then applying a single flip $H$ gives a sign $(-1)^{|H|-1}$. Similarly, the winding number changes by $|H|-1$, giving the same contribution. In the final step, a potential orientation flip would contribute $(-1)^{n}$, while the winding number changes by $n-2$, thus giving an identical sign contribution. This shows that $(-1)^{w(\beta|\gamma)+1}$ calculates the correct sign of \eqref{theorem-1-claim} and hence concludes the proof.
\end{proof}

\begin{example}
Let us illustrate this procedure in practice by calculating $\la \C(12345), \C(13245) \ra$, or equivalently $m_{\alpha'}(12345|13245)$. After drawing a circle diagram with both permutations $(12345)$ and $(13245)$, we find the dual of the polygon created by the red loop. This results in a diagram connecting two subamplitudes, $m_{\alpha'}(23,-I|23,-I)$ and $m_{\alpha'}(451I | 451I)$ with $k_I^{\mu} = k_2^\mu + k_3^\mu$, with the propagator given by $1/\sin \pi s_{23}$. The first subamplitude is equal to $1$, while the second one is a sum of two propagators given by tangents in the factorization channels $s_{34}$ and $s_{51}$, according to \eqref{C-1234-1234}. To summarize, we have:
\be
\includegraphics[valign=c]{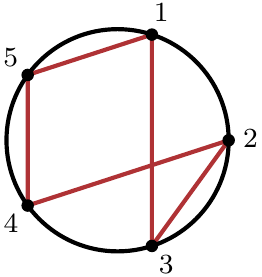} \;=\; \includegraphics[valign=c]{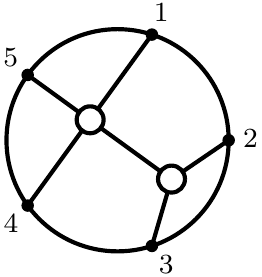} \;=\; -\left(\frac{i}{2}\right)^{2} \frac{1}{\sin \pi s_{23}} \left( \frac{1}{\tan \pi s_{45}} + \frac{1}{\tan \pi s_{51}} \right).
\ee
The minus sign arises because of the winding number $w(12345|13245) = 2$. The above result encodes the fact that the two associahedra $K_4(12345)$ and $K_4(13245)$ share the face $(23)$, as well as two vertices $(23)\cap (123)$ and $(23) \cap (234)$, see Figure~\ref{fig-blowup}.
\end{example}

\begin{example}
Let us consider another example by evaluating $\la \C(12345), \C(12453) \ra$, or equivalently $m_{\alpha'}(12345|12453)$. This time, the dual graph is composed of three subamplitudes, each of which is trivalent and hence contributes the contact term $1$. The two propagators in the $s_{12}$ and $s_{45}$ channels are given by sine factors, as follows:
\be
\includegraphics[valign=c]{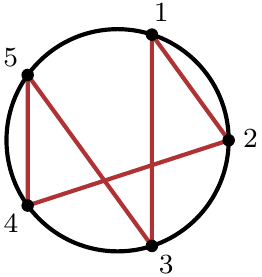} \;=\; \includegraphics[valign=c]{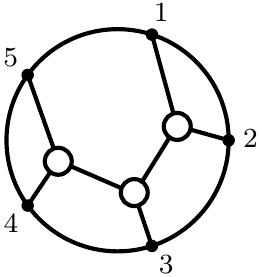} \;=\; -\left(\frac{i}{2}\right)^{2} \frac{1}{\sin \pi s_{12}}\, \frac{1}{\sin \pi s_{45}}.
\ee
The winding number is $w(12345|12453) = 2$, giving an overall minus sign. Since there is only one trivalent graph contributing, it means that the associahedra $K_4(12345)$ and $K_4(12453)$ share only a single vertex $(12) \cap (123)$, see Figure~\ref{fig-blowup}.
\end{example}

\begin{example}
Next, let us consider a six-point example of $\la \C(123456), \C(124365) \ra$, or equivalently $m_{\alpha'}(123456|124365)$. The dual diagram reduces to three subamplitudes, two of which give $1$, while the third contributes a four-point self-intersection number \eqref{C-1234-1234}. Recall that self-intersection numbers themselves are given by diagrams with vertices of odd valency and propagators built out of tangent functions \cite{Mizera:2016jhj}. In our case, we have:
\be
\includegraphics[valign=c]{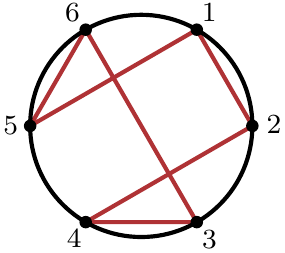} \;=\; \includegraphics[valign=c]{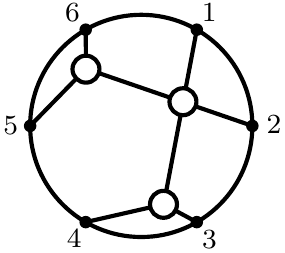} \;=\; \left(\frac{i}{2}\right)^{3} \frac{1}{\sin \pi s_{34}}\, \frac{1}{\sin \pi s_{56}} \left( \frac{1}{\tan \pi s_{12}} + \frac{1}{\tan \pi s_{234}} \right).
\ee
The plus sign arises because $w(123456|124365)=3$. The above answer also encodes the fact that the two associahedra $K_5(123456)$ and $K_5(124365)$ share the vertices $(34) \cap (1234) \cap (12)$ and $(34) \cap (1234) \cap (234)$, as well as the codimension-$2$ edge $(34) \cap (1234)$ between them.
\end{example}

\begin{example}
Finally, let us consider a $12$-point example of the intersection number of twisted cycles $ \tC(123456789,10,11,12)$ and $\tC(127856,11,12,9,10,34)$. For these two permutations the dual diagram is trivalent and hence given as a product of propagators of the sine type as follows:
\begin{align}
\includegraphics[valign=c]{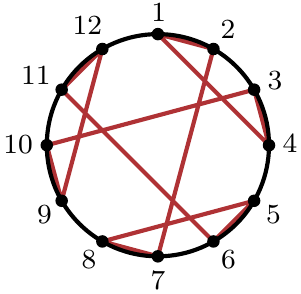} \;=\; \includegraphics[valign=c]{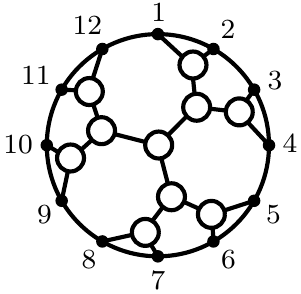} \;=\;& -\left(\frac{i}{2}\right)^{9} \frac{1}{\sin \pi s_{12}}\, \frac{1}{\sin \pi s_{34}}\, \frac{1}{\sin \pi s_{56}}\, \frac{1}{\sin \pi s_{78}}\, \frac{1}{\sin \pi s_{9,10}}\tr[-3em]
&\quad\;\;\times \frac{1}{\sin \pi s_{11,12}}\, \frac{1}{\sin \pi s_{1234}}\, \frac{1}{\sin \pi s_{5678}}\, \frac{1}{\sin \pi s_{9,10,11,12}}.
\end{align}
The winding number equals to $4$, which gives an overall minus sign. The corresponding nine-dimensional associahedra intersect only at a single vertex $(12) \cap (34) \cap (56) \cap (78) \cap (9,10) \cap (123456789,10) \cap (1234) \cap (5678) \cap (12345678)$ in the moduli space. For more examples of how to evaluate $m_{\alpha'}(\beta | \gamma)$ we refer the reader to \cite{Mizera:2016jhj}.
\end{example}

In the field theory limit, $\alpha' \to 0$, only the faces of maximal codimension contribute. In other words, a given intersection number reduces to a sum over trivalent diagrams. Since the intersection number has only terms coming from the intersection face, $F = K_{n-1}(\beta) \cap K_{n-1}(\gamma)$, these diagrams are necessarily planar with respect to both orderings, $\beta$ and $\gamma$. This gives rise to the usual definition of the bi-adjoint scalar double-partial amplitudes $m(\beta | \gamma)$ \cite{BjerrumBohr:2012mg,Cachazo:2013iea}.

In general, since each facet of a given associahedron can be written as a product of lower-dimensional associahedra according to \eqref{hyperface-isomorphism}, the corresponding scattering amplitude factors into two lower-point amplitudes connected by a label $I$. Physically, this is the statement of \emph{unitarity}. Since the regularized twisted cycles are given by the generalized Pochhammer contour, near the facet it receives the contribution of $1/(e^{2\pi i s_I} - 1)$, which contains an infinite number of simple poles at $s_I = 0, \pm 1, \pm 2, \ldots$, allowing propagation of massless, massive, and tachyonic modes. Physically, this corresponds to the statement of \emph{locality}. Both of these properties are associated to twisted cycles. In the case of pairings between twisted cycles and twisted cocycles, such as the ones giving rise to open string amplitudes \eqref{expansion}, it is the role of the cocycle to select if a given factorization channel is utilized or not.

\pagebreak
\section{\label{sec-conclusion}Conclusion}

\textsc{In this work we have shed} new light on the Kawai--Lewellen--Tye relations \cite{Kawai:1985xq}. By applying the tools of twisted de Rham theory, we have shown that they follow from the underlying algebro-topological identities known as the \emph{twisted period relations}. On the way, we have formulated tree-level string theory amplitudes in a way that makes connections to combinatorics and topology. In particular, we have explored the relation to the polytope called the \emph{associahedron}. We have shown that the inverse of the KLT kernel can be computed from the knowledge of how associahedra intersect one another in the moduli space. From this perspective, the \emph{inverse} of the KLT kernel appears to be a more fundamental object than the kernel itself, in both string and field theory. Introduction of twisted de Rham theory in the study of string integrals opens new directions not only for the KLT relations, but also scattering amplitudes in a more general setting.

Since the formalism of twisted de Rham theory applies to a broad spectrum of topological spaces and multi-valued functions, one may wonder about generalizations of the calculations presented in this work to other cases. Indeed, the most natural extension is to consider higher-genus amplitudes in string theory. This direction looks particularly promising in the light of the recent analysis of monodromy properties of higher-genus string integrals by Tourkine and Vanhove \cite{Tourkine:2016bak}, as well as Hohenegger and Stieberger \cite{Hohenegger:2017kqy}. A related construction of KLT relations for one-loop field theory integrands has been recently given by He, Schlotterer, and Zhang \cite{He:2016mzd,He:2017spx}. We leave the study of these intriguing questions for future research. Aside of string theory, intersection numbers of twisted cycles have been previously calculated in the context of conformal field theories by Mimachi and Yoshida \cite{Mimachi2003,Mimachi2004} in order to explain identities between correlation functions found by Dotsenko and Fateev \cite{Dotsenko:1984ad,Dotsenko:1984nm}. It would be interesting to investigate whether these identities can be further generalized, perhaps even to the case of the conjectured conformal field theory on the null boundary of asymptotically-flat spacetimes \cite{Strominger:2017zoo}.

The Deligne--Mumford--Knudsen compactification of the moduli space $\mathcal{M}_{0,n}$ can be constructed as a Chow quotient of the Grassmannian $\text{Gr}(2,n)$ \cite{1992alg.geom.10002K}. It is also known that planar amplitudes in the $\mathcal{N}=4$ super Yang--Mills (SYM) theory---which are the field theory limit of superstring amplitudes---can be defined on the positive Grassmannian \cite{ArkaniHamed:2012nw}. Additionally, they also have a description in terms of a geometric object found by Arkani-Hamed and Trnka \cite{Arkani-Hamed:2013jha,Arkani-Hamed:2013kca}. It is natural to expect it to be related to the associahedron described in this work. The problem is somewhat akin to the so-called Grassmanian dualities \cite{ArkaniHamed:2009dg} with an additional complication of the $\alpha'$ corrections. The main obstacle comes from the fact that our results are largely independent of the spacetime dimension and precise theory under consideration, as it only suffices that its amplitudes have a BCJ representation in terms of Z-integrals. Associahedra also appear in a slightly different context of cluster algebras used to study $\mathcal{N}=4$ SYM amplitudes \cite{Golden:2013xva,Golden:2014xqa}.\footnote{Some steps in the direction of applying twisted de Rham theory to $\mathcal{N}=4$ SYM amplitudes have been taken in \cite{Abe:2015ucn}.}

One may then wonder about seeing the field theory limit from a different point of view. Indeed, one such possibility coming naturally from twisted de Rham theory is to consider the dual twisted homology and cohomology defined with the multi-valued function $u^{-1}(z)$ instead of $\overbar{u(z)}$. This corresponds to reversing the string parameter $\alpha'$ in one of the amplitudes. An equivalent procedure has been recently studied in the context of string theory by Siegel and collaborators \cite{Siegel:2015axg,Huang:2016bdd,Leite:2016fno,Li:2017emw}. In particular, Huang, Siegel, and Yuan showed that considering a \emph{chiral} version of KLT relations, one obtains the field theory amplitude as follows \cite{Huang:2016bdd}:
\be\label{chiral-KLT-relation}
\mathcal{A}^{\text{GR}} = \!\!\sum_{\beta \in \mathcal{B},\, \gamma \in \mathcal{C}}\!\! \mathcal{A}^{\text{open}}(\beta)\, m_{\alpha'}^{-1}(\beta | \gamma)\, \mathcal{A}^{\text{open}}_{\text{chiral}}(\gamma).
\ee
where we have used the same notation as in \eqref{KLT-relation}, except $\mathcal{A}^{\text{open}}_{\text{chiral}}(\gamma)$ denotes a string amplitude with a flipped spacetime signature $\eta_{\mu\nu} \to -\eta_{\mu\nu}$, or equivalently a replacement $\alpha' \to -\alpha'$. The equation \eqref{chiral-KLT-relation} is in fact equivalent to the original twisted Riemann period relation found in \cite{cho1995}. In this interpretation, the field theory amplitude $\mathcal{A}^{\text{GR}}$ is an intersection number of twisted \emph{cocycles}, given by the pairing:
\be
H^{n-3}_c(X,\mathcal{L}_\omega) \times H^{n-3}(X,\mathcal{L}_\omega^\vee) \;\longrightarrow\; \mathbb{C},
\ee
where the dual system is defined through $\mathcal{L}_\omega^\vee = \mathcal{L}_{-\omega}$ and one of the cohomologies is with compact support \cite{cho1995}. Methods for evaluation of this pairing have been given in \cite{zbMATH03996010,cho1995,cho-private-note,matsumoto1998,Matsumoto1998-2,Ohara98intersectionnumbers}. In particular, Matsumoto gave a simple proof \cite{matsumoto1998} of the fact that these intersection numbers localize on the intersections of the singular locus of $u(z)$. Translating to the string theory case, with the bases of twisted cocycles given by \eqref{PT-def}, the intersection numbers are given by a sum of multi-dimensional residues around the vertices of \emph{all} the associahedra in the moduli space $\widetilde{\mathcal{M}}_{0,n}$. A given vertex contributes if only if the differential form $\PT(\beta) \w \overbar{\PT(\gamma)}$ has a double pole at the place corresponding to this vertex. Since both differential forms are logarithmic, double poles arise only if the two Parke--Taylor factors share a factorization channel. In this way, the sum over all residues receives contributions only from the vertices laying on the intersection $K_{n-1}(\beta) \cap K_{n-1}(\gamma)$ in the moduli space. This gives rise to the bi-adjoint scalar amplitude $m(\beta | \gamma)$ as a sum over all Feynamn diagrams compatible with both permutations $\beta$ and $\gamma$. It is also possible to consecutively apply the global residue theorem (GRT) \cite{griffiths2014principles}---in a way analogous to the one considered by Dolan and Goddard \cite{Dolan:2013isa}---in order to obtain a dual description that localizes on the residues around the \emph{scattering equations} $\bigwedge_{i=2}^{\!n-2} \{|E_i| = \varepsilon\}$. In pictures, the duality translates between different types of residues as follows:
\be
\includegraphics[scale=1]{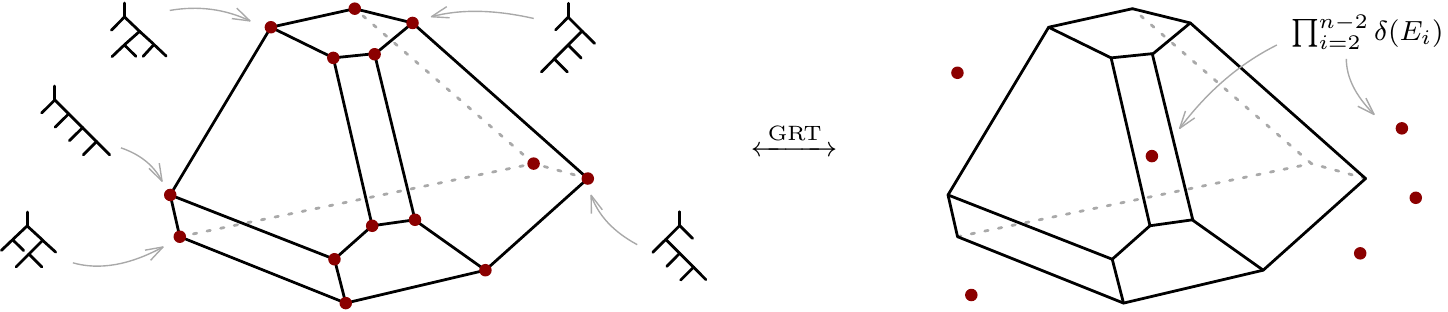}
\ee
It is known that when all the exponents of the Koba--Nielsen factor \eqref{Koba-Nielsen} are positive, there are $(n-3)!$ solutions of the constraint $\prod_{i=2}^{n-2} \delta(E_i)$ laying in the hypercube $(0,1)^{n-3} \subset \widetilde{\mathcal{M}}_{0,n}(\mathbb{R})$ with one solution per associahedron \cite{Cachazo:2016ror}. Of course, the advantage of this approach is the fact that the residues are computed far away from the faces, so no blowup is necessary for explicit computations. The resulting formula for the intersection number of twisted cocycles reads, up to normalization factors:
\be\label{CHY}
\oint_{\bigwedge_{i=2}^{\!n-2} \{|E_i| = \varepsilon\}} \!\!\! \frac{\PT(\beta) \w \overbar{\PT(\gamma)}}{\prod_{i=2}^{n-2} E_i}
\ee
This is the so-called Cachazo--He--Yuan formula \cite{Cachazo:2013hca,Cachazo:2013iea} for the bi-adjoint scalar amplitude $m(\beta | \gamma)$. Other amplitudes, such as $\mathcal{A}^{\text{GR}}$, can be expanded in the basis of $m(\beta | \gamma)$, as in \eqref{expansion}. Notice that the field theory limit is obtained in the limit of vanishing twist $\omega = \alpha' \sum_{i=2}^{n-2} E_i\, dz_i$. It is tempting to suggest that there could exist an alternative derivation of the CHY formulae by imposing a constraint on $\omega$, perhaps in relation to Morse theory discussed in \cite{aomoto2011theory}. It would also be interesting to find out how it relates to other approaches of connecting CHY formalism and string theory, see, e.g., \cite{Berkovits:2013xba,Bjerrum-Bohr:2014qwa,Casali:2016atr,He:2016iqi,Li:2017emw,Mizera:2017sen}, in particular in the context of ambitwistor strings \cite{Mason:2013sva}. We leave the study of these connections for future investigations.

Also in the field theory limit, there exists another instance of relations between gravity and Yang--Mills amplitudes known as the \emph{BCJ double-copy} \cite{Bern:2008qj}. In this context, Carrasco studies the space of trivalent graphs and its relation to associahedra and permutohedra \cite{Carrasco-talk}. It would be interesting to see how this story fits with ours. Here, twisted cycles play the role of \emph{colour} factors, while twisted cocycles play the role of \emph{kinematics} factors. Moreover, blowup of the moduli space $\widetilde{\mathcal{M}}_{0,n}$---or its double cover \cite{doi:10.1137/130947532}---provides a natural way of understanding the configuration space of trivalent diagrams as its limit. We hope this language could contribute to deeper understanding of \emph{colour-kinematics} duality and its connection to KLT relations, particularly at higher loops.

Last but not least, it is important to understand questions arising from this work on the level of rigour of mathematics. Such issues involve, for example, study of the Hodge structure of the intersection form for twisted cohomology groups \eqref{intersection-form}, finding an algebro-topological derivation of the form of the circuit matrix given in \eqref{circuit-matrix}, or study of the formula \eqref{CHY} in the context of Morse theory. From the point of view of combinatorics, further study of the moduli spaces of marked bordered Riemann surfaces of higher genus and their tilings, along the lines of \cite{Liu:587181,devadoss2011deformations}, is important in understanding higher-loop generalizations of KLT relations. It is also known that string amplitudes have a rich motivic structure, see, e.g., \cite{Schlotterer:2012ny,Broedel:2013tta,Stieberger:2016xhs}. In particular, J-integrals \eqref{J-intergral} and Z-integrals \eqref{Z-theory-amplitude} can be related by \cite{Stieberger:2014hba}:
\be
J(\beta | \gamma) = \textsf{sv}\left[ Z_\beta(\gamma) \right],
\ee
where $\textsf{sv}$ is the \emph{single-valued projection} introduced by Brown \cite{BROWN2004527,Brown:2013gia}. This relation bears resemblance to the twisted period relations for the basis of twisted cycles and cocycles \eqref{twisted-period-relations-Z}. It would interesting to study the connection between motivic structure and twisted de Rham theory in this context.

\begin{appendices}

\titleformat{name=\section}[display]
{\normalfont}
{\footnotesize\textsc{Appendix \thesection}}
{0pt}
{\Large\bfseries}
[\vspace{-10pt}\color{Maroon}\rule{\textwidth}{.6pt}]

\pagebreak
\section{\label{app-field-theory-limit}Field Theory Limit from the Generalized Pochhammer Contour}

\textsc{As was shown} in Section~\ref{sec-klt-as-twisted-period-relations}, tree-level open string partial amplitudes can be understood as pairings between twisted cycles and twisted cocycles. In this appendix, we show how to obtain its field theory limit, $\alpha' \to 0$, by utilizing the generalized Pochhammer contour and blowup of the moduli space described in Section~\ref{subsec-regularization}.

Recall that open string amplitudes can be expanded in the basis of Z-theory amplitudes \eqref{Z-theory-amplitude}. They in turn are given by the pairing:
\be\label{cycle-cocycle-pairing}
H_{n-3}(X,\mathcal{L}_\omega) \times H^{n-3}(X,\nabla_\omega) \;\longrightarrow\; \mathbb{C},
\ee
denoted by $Z_{\beta}(\gamma) = \la \C(\beta), \PT(\gamma) \ra$, using the basis of twisted cycles \eqref{string-cycles} and cocycles \eqref{PT-def} for string amplitudes. We also employ the regularization $\reg\, \tC(\beta)$ in order to make twisted cycles compact, and work on the blowup of the moduli space, $\widetilde{\M}_{0,n}$. Notice that the information about the factors of $\alpha'$ of $Z_{\beta}(\gamma)$ is entirely contained in the regularized twisted cycle $\reg\, \tC(\beta)$. In order to take the field theory limit, let us count the powers of $\alpha'$ contributing to different pieces of the generalized Pochhammer contour based on the associahedron $K_{n-1}(\beta)$.

Each face $F$ of codimension $k$ can be written as $F = H_1 \cap H_2 \cap \cdots \cap H_k$. Near each facet $H_i$, $\reg\, \tC(\beta)$ picks up a factor $1/(e^{2\pi i \alpha' s_{H_i}} - 1)$, which in the $\alpha' \to 0$ limit scales as $1/\alpha'$. We conclude that the string integral in the $\alpha' \to 0$ limit receives leading contributions from the faces $F$ of maximal codimension $n-3$, or in other words, vertices of the associahedron $K_{n-1}(\beta)$. Since all the singularities of the string amplitude are encapsulated in the choice of the generalized Pochhammer contour, and the integrals to be performed are finite when $\alpha' \to 0$, we can take  $u(z) \to 1$ in the same limit. To summarize, we have:
\be\label{field-theory-limit}
\lim_{\alpha' \to 0} \la \C(\beta), \PT(\gamma) \ra \;=\; \frac{1}{(2\pi i \alpha')^{n-3}} \! \sum_{v = H_1 \cap \cdots \cap H_{n-3}} \frac{1}{\prod_{i=1}^{n-3} (\pm s_{H_i})} \!\! \oint\limits_{\substack{|H_{i}|=\varepsilon\\ i=1,\ldots,n-3}} \!\!\!\!\PT(\gamma).
\ee
where the sum proceeds over all the Catalan number $C_{n-2}$ \cite{A000108} of vertices $v$ of $K_{n-1}(\beta)$. The integrals are performed along an appropriately oriented tubular neighbourhood of each vertex $v$.\footnote{For a reference on the computation of multi-dimensional contour integrals see, e.g., Chapter 5 of \cite{griffiths2014principles}.}

Let us work out explicit examples of the evaluation of \eqref{field-theory-limit}. One needs to be extra careful about sign factors coming from orientation induced by the associahedron on the vertices. Let us illustrate this fact for $n=4$. In the $\alpha' \to 0$ limit, the regularized twisted cycle defined in \eqref{reg-0-1} becomes
\be
\lim_{\alpha' \to 0} \reg \overrightarrow{(0,1)} = \frac{1}{\alpha'} \left( \frac{S(\varepsilon,0)}{2\pi i s} - \frac{S(1-\varepsilon,1)}{2\pi i t} \right).
\ee
Recall that we use $S(a,z)$ to denote a positively-oriented circular contour starting at $a$ and with a centre at $z$. The contours around the two vertices of $K_{3}(\I_4)$ come with different signs due to different orientations induced from $\overrightarrow{(0,1)}$. Let us now evaluate \eqref{field-theory-limit} for a four-point example $\la \C(1234), \PT(1234) \ra$. From the pole around $z=0$ we obtain:
\be\label{C-1234-PT-1234-a}
\lim_{\alpha' \to 0} \la \C(1234), \PT(1234) \ra \bigg|_{z=0} = \frac{1}{2\pi i \alpha' s} \oint_{|z| = \varepsilon} \frac{dz}{(0-z)(z-1)} = \frac{1}{\alpha' s},
\ee
and from around $z=1$ we find the contribution:
\be\label{C-1234-PT-1234-b}
\lim_{\alpha' \to 0} \la \C(1234), \PT(1234) \ra \bigg|_{z=1} = -\frac{1}{2\pi i \alpha' t} \oint_{|z-1| = \varepsilon} \frac{dz}{(0-z)(z-1)} = \frac{1}{\alpha' t}.
\ee
Hence we find the answer which is a sum over two Feynman diagrams in the $s$ and $t$ channels. The two contributions worked out to give the same sign. In general, all the vertices contributing to \eqref{field-theory-limit} will give the same sign. For another choice of the twisted cocycle, $\PT(2134)$, we have:
\be
\lim_{\alpha' \to 0} \la \C(1234), \PT(2134) \ra = \frac{1}{2\pi i \alpha' s} \oint_{|z| = \varepsilon} \frac{dz}{(z-0)(0-1)} = -\frac{1}{\alpha' s}.
\ee
Notice that contribution from the vertex $z=1$ vanishes, since $\PT(2134)$ does not have a pole at $z=1$. Similarly, for $\PT(1324)$ we obtain:
\be
\lim_{\alpha' \to 0} \la \C(1234), \PT(1324) \ra = -\frac{1}{2\pi i \alpha' t} \oint_{|z-1| = \varepsilon} \frac{dz}{(0-1)(1-z)} = -\frac{1}{\alpha' t},
\ee
since there are no poles at $z=0$.

For higher-point cases one needs to consider a blowup of the integrals \eqref{field-theory-limit}. Let us illustrate the procedure with an $n=5$ example for $\la \C(12345), \PT(12345) \ra$. The corresponding associahedron $K_{4}(\I_5)$ has five vertices. The contribution from $(z_2, z_3) = (0,1)$ can be calculated straightforwardly:
\begin{align}
\lim_{\alpha' \to 0} \la \C(12345), \PT(12345) \ra \bigg|_{(12) \cap (34)} &= -\frac{1}{(2\pi i \alpha')^2}\frac{1}{s_{12}\, s_{34}} \oint_{\substack{|z_2| = \varepsilon\\ |z_3 - 1| = \varepsilon}} \frac{dz_2 \w dz_3}{(0-z_2)(z_2 - z_3)(z_3 - 1)}\tr
&= -\frac{1}{\alpha'^2\, s_{12}\, s_{34}}.
\end{align}
Here we have used the tubular contour given by $\{ |z_2| = \varepsilon \} \w \{ |z_3 - 1| = \varepsilon\}$. Next, near $(z_2, z_3) = (0,0)$ we perform a blowup using the change of variables from $\{z_2, z_3\}$ into $\{y_2, \tau\}$ given by $z_2 = \tau y_2$ and $z_3 = \tau$. Since $dz_2 \w dz_3 = \tau dy_2 \w d\tau$, we have the contribution:
\begin{align}
\lim_{\alpha' \to 0} \la \C(12345), \PT(12345) \ra \bigg|_{(123)} &= \frac{1}{(2\pi i \alpha')^2} \sum_{H \in \{(12), (23)\}} \frac{1}{s_{123} (\pm s_H)} \oint_{\substack{|H| = \varepsilon \\ |\tau| = \varepsilon} } \frac{\tau dy_2 \w d\tau}{(0-\tau y_2)(\tau y_2 - \tau)(\tau - 1)}\tr
&= -\frac{1}{2\pi i \alpha'^2} \sum_{H \in \{(12), (23)\}} \frac{1}{s_{123} (\pm s_H)} \oint_{|H| = \varepsilon} \frac{dy_2}{(0-y_2)(y_2 - 1)}\tr
&= - \frac{1}{\alpha'^2 s_{123}}\left( \frac{1}{s_{12}} + \frac{1}{s_{23}}\right).
\end{align}
In the first line, the powers of $\tau$ add up to create a simple pole $d\tau / \tau$, over which we have integrated. In the second line we have used the result of the four-point computations \eqref{C-1234-PT-1234-a} and \eqref{C-1234-PT-1234-b} with the appropriate signs for $s_H$, $H \in \{(12),(23)\}$. For the remaining two vertices near $(z_2,z_3) = (1,1)$ we use the variables ${\tau, y_3}$ defined through $z_2 = 1 - \tau$ and $z_3 = 1 - \tau y_3$, so that $dz_2 \w dz_3 = \tau d\tau \w dy_3$. A similar calculation reveals:
\begin{align}
\lim_{\alpha' \to 0} \la \C(12345), \PT(12345) \ra \bigg|_{(234)} &= \frac{1}{(2\pi i \alpha')^2} \sum_{H \in \{(23), (34)\}} \frac{1}{s_{234} (\pm s_H)} \oint_{\substack{|\tau| = \varepsilon \\ |H| = \varepsilon} } \frac{\tau d\tau \w dy_3}{(-1+\tau)(-\tau + \tau y_3)(-\tau y_3)}\tr
&= -\frac{1}{2\pi i \alpha'^2} \sum_{H \in \{(23), (34)\}} \frac{1}{s_{234} (\pm s_H)} \oint_{|H| = \varepsilon} \frac{dy_3}{(0-y_3)(y_3 - 1)}\tr
&= - \frac{1}{\alpha'^2 s_{234}}\left( \frac{1}{s_{23}} + \frac{1}{s_{34}}\right),
\end{align}
where once again we have used a residue theorem to integrate over the simple pole $d\tau / \tau$. Summing up all the contributions and using momentum conservation, we have
\be
\lim_{\alpha' \to 0} \la \C(12345), \PT(12345) \ra = -\frac{1}{\alpha'^2} \left( \frac{1}{s_{12}s_{34}} + \frac{1}{s_{23}s_{45}} + \frac{1}{s_{34}s_{51}} + \frac{1}{s_{45}s_{12}} + \frac{1}{s_{51}s_{23}} \right).
\ee
Using the same procedure with different cocycles, it is straightforward to verify other examples, for instance:
\be
\lim_{\alpha' \to 0} \la \C(12345), \PT(13245) \ra = \frac{1}{\alpha'^2 s_{23}} \left( \frac{1}{s_{45}} + \frac{1}{s_{51}}\right),
\ee
\be
\lim_{\alpha' \to 0} \la \C(12345), \PT(12453) \ra = \frac{1}{\alpha'^2 s_{12} s_{45}}.
\ee

In general, in the $\alpha' \to 0$ limit one finds that Z-integrals \eqref{Z-theory-amplitude} collapse to the bi-adjoint scalar partial-amplitudes \cite{Mafra:2016mcc}:
\be
\lim_{\alpha' \to 0} \la \C(\beta), \PT(\gamma) \ra = - (-\alpha')^{3-n}\, m(\beta | \gamma),
\ee
where we have included the normalization factor. The method of computing this result using the generalized Pochhammer contour presented above, despite having a simple geometrical interpretation in terms of the associahedron, is not particularly efficient. In this light, it would be interesting to study systematic ways of evaluating \eqref{field-theory-limit} and its higher-order terms, which could provide a new way of performing the $\alpha'$ expansion.

\end{appendices}

\pagebreak
\bibliographystyle{JHEP}
\bibliography{references}

\end{document}